%% file: main.tex
\documentclass[11pt]{article}
\input{pkgs.tex}
\usepackage[capitalise]{cleveref}

\input{defs.tex}
\usepackage{cite}
\usepackage{balance}
\usepackage[margin=1.0in]{geometry}
\usepackage[toc,page]{appendix}

\begin{document}
\pagestyle{fancy}
\title{Global Convergence of Policy Gradient for Sequential \\Zero-Sum Linear Quadratic Dynamic Games}
\author{Jingjing Bu \and Lillian J. Ratliff \and Mehran Mesbahi\thanks{The authors are with the University of Washington, Seattle, WA, 98195; Emails: \{bu,ratliffl,mesbahi\}@uw.edu}}
\date{}
\maketitle

\begin{abstract}
  We propose projection-free sequential algorithms for linear-quadratic dynamics games. These policy gradient based algorithms are akin to Stackelberg leadership model and can be extended to model-free settings. We show that if the ``leader'' performs natural gradient descent/ascent, then the proposed algorithm has a global sublinear convergence to the Nash equilibrium. Moreover, if the leader adopts a quasi-Newton policy, the algorithm enjoys a global quadratic convergence. Along the way, we examine and clarify the intricacies of adopting sequential policy updates for LQ games, namely, issues pertaining to stabilization, indefinite cost structure, and circumventing projection steps.
\end{abstract}

\section{Introduction}
\label{sec:intro}

Linear-quadratic (LQ) dynamic and differential games exemplify situations where
two players influence an underlying linear dynamics in order to respectively, minimize and maximize a given quadratic cost on the state and the control over an infinite time-horizon.\footnote{We will adopt the convention of referring to the continuous time scenario as differential games. Moreover, in this paper, we focus on infinite horizon LQ games without a discount factor.} This LQ game setup has a rich history in system and control theory, not only due to its wide range of applications but also since it directly extends its popular one player twin, the celebrated linear quadratic regular (LQR) problem~\cite{Bernhard2013-aa,Engwerda2005-np,Zhang2005-pe}.
LQR on the other hand, is one of the foundations of modern system theory~\cite{willems1971least, Anderson2007-mn}. This partially stems from the fact that the elegant analysis of minimizing a quadratic (infinite horizon) cost over an infinite dimensional function space leads to a solution that is in the constant feedback form, that can be obtained via solving the \emph{algebraic Riccati equation} (ARE)~\cite{Lancaster1995algebraic}. 
As such, solving LQR and its variants have often been approached from the perspective
of exploiting the structure of ARE~\cite{bini2012numerical}.
The ARE facilitates solving for the so-called cost-to-go (encoded by a positive semi-definite matrix), that
can subsequently be used to characterize the optimal state feedback gain.
This general point of view has also influenced the ``data-driven'' approaches for solving
the generic LQR problem. For instance, in the value-iteration for 
reinforcement learning (RL)---e.g., $Q$ learning---one aims to first estimate 
the cost-to-go at a given time instance and through this estimate, update the state feedback gain.
In recent years, RL has witnessed major advances for wide range of 
decision-making problems (see, e.g.,~\cite{silver2016mastering,mnih2015human}). 

Direct policy update is another algorithmic pillar
for decision-making over time.
The conceptual simplicity of policy optimization offers advantages in terms of
computational scalability and potential extensions to model-free settings. 
As such, there is a renewed interest in analyzing the classical LQR problem under the RL framework from the perspective of direct policy updates~\cite{dean2017sample, fazel2018global}.
The extension of LQR control to multiple-agent settings, i.e., LQ dynamic and differential 
games, has also been explored by the  game theory and optimal control
communities~\cite{basar1999dynamic,Vamvoudakis2017-ls}. 
Two-person zero-sum LQ dynamic and differential games
are particular instances of this more general setting where two players aim to optimize an objective with
directly opposing goals subject to a shared linear dynamical system. 
More
precisely, one player attempts to minimize the objective while the other aims to
maximize it. This framework has important applications in $\ca H_{\infty}$
optimal control~\cite{basar2008h}. In fact, as it is well-known in control theory,
the saddle point solution---i.e., a Nash equilibrium (NE)---of an
infinite horizon LQ game can be obtained via the corresponding \emph{generalized algebraic
Riccati equation} (GARE). As such, seeking the NE in a zero-sum LQ dynamic
game revolves around solving the 
GARE~\cite{stoorvogel1994discrete}. Recently, multi-agent
RL~\cite{srinivasan2018actor, jaderberg2019human} has also achieved impressive
performance by using direct policy gradient updates. Since LQ dynamic games have
explicit solutions via the GARE, understanding the performance of policy gradient
based algorithms for LQ games could serve as a benchmark, and providing deeper insights into theoretical
guarantees of  multi-agent RL in more general settings~\cite{Vamvoudakis2017-ls}. 
In the meantime, the application of policy optimization algorithms in the game
setting proves to require more intricate analysis due to the fact that the infinite horizon cost is undiscounted and potentially unbounded per stage.
As such, it is well known that devising direct policy iterations for undiscounted and unbounded per stage cost functions in the RL setting is nontrivial~\cite{bertsekas2005dynamic}. The cost structure of standard LQR, however, streamlines the design of policy based iterations~\cite{hewer1971iterative, fazel2018global, bu2019lqr}.\footnote{In~\cite{fazel2018global, bu2019lqr} it has been assumed that the quadratic state cost is via a positive definite $Q$; this is not ``standard'' as only the detectability of the pair $(Q, A)$ and $Q \succeq 0$ is required for LQR synthesis.} Nevertheless in policy iteration, special care has to be exercised to ensure that the iterative policy updates are in fact stabilizing.
%
The stabilization issues is particularly relevant in the LQ dynamic games.
Note that the policy space for LQ games is an open set admitting a cartesian product structure. 
Hence, in the policy updates for LQ dynamic games (say in the RL setting), we must guarantee that the iterates jointly stay in the open set as otherwise the ``simulator'' would diverge. Recently, Zhang \emph{et al.}~\cite{zhang2019policy}, using certain assumptions and relying on a projection step, have proposed a sequential direct policy updates with a sublinear convergence for LQ dynamic games. 
\paragraph{{\bf Contributions.}} In this paper, we first clarify the setting for
discussing sequential LQ dynamic games, particularly addressing issues pertaining to stabilization. 
We then propose leader-follower type algorithms that resemble the Stackelberg leadership model~\cite{basar1999dynamic}. 
Specifically, in the proposed iterative algorithms for LQ games, one player is designated as a leader 
and the other as a follower. 
We require that the leader plays natural gradient or quasi-Newton policies while the follower can play any first-order based policies. In particular, we do not require a specific player to be the leader as long as the algorithm can be initialized appropriately. We prove that if the leader performs a natural gradient policy update, then the proposed leader-follower algorithm has a global sublinear convergence and asymptotic linear convergence rate. In the meantime, if the leader adopts a quasi-Newton policy update, the algorithm converges at a global quadratic rate.
Moreover, we show that gradient policy (respectively, natural gradient 
and quasi-Newton policies) has a global linear (respectively, linear and
quadratic) convergence to the optimal stabilizing feedback gains even when the
state cost matrix $Q$ is indefinite. This result essentially extends
the results on standard LQR investigated in~\cite{fazel2018global, bu2019lqr},
where the analysis relies on the assumption of having $Q \succ 0$; this
extension is of independent interest for various optimal control
applications as well, e.g., control with conflicting
objectives~\cite{willems1971least}). Compared with the results presented
in~\cite{zhang2019policy}, the contributions of this work include
the following:
\begin{enumerate}
  \item[(1)] We remove the ``nonstandard assumption'' that the NE point $(K_*, L_*)$ satisfies $Q-L_*^{\top}R_2 L_* \succ 0$.\footnote{It is noted that in~\cite{zhang2019policy} that the projection step is not generally required in implementations; however the convergence analysis presented in~\cite{zhang2019policy} is based on such a projection.} We note that such an assumption is not  standard in the LQ literature and as such, needs further justification beyond its algorithmic implications.  In fact, in the analysis presented in~\cite{zhang2019policy}, it is crucial for the convergence of the  algorithm to a priori know the positive number $\varepsilon > 0$ for which $\{L: Q - L_*^{\top}R_2 L_* \succeq \varepsilon I\}$; moreover, one has to be able to project onto this set.  
  \item[(2)] Our setting allows for larger stepsizes for the policy iteration in LQ dynamic games, greatly improving its practical performance. This is facilitated by providing insights into the stabilizing policy updates through a careful analysis of the corresponding indefinite GARE.
\item[(3)] We clarify the interplay between key concepts in control and optimization in the convergence analysis of the proposed iterative algorithms for LQ dynamic games. This is inline with our belief that identifying the role of concepts such stabilizability and detectability in the convergence analysis of ``data-guided" algorithms for decision making problems  with an embedded dynamic system is of paramount importance.
  \item[(4)] We show that the quasi-Newton policy has a global quadratic convergence rate for LQ dynamic games. This result might be of independent interest for discrete-time GARE, data-driven or not. To the best of our knowledge, the proposed algorithm is the first iterative approach for discrete-time GARE with a global quadratic convergence.
    \item[(5)] Finally, we show that in the proposed iterative algorithms for LQ dynamic games, any player can assume the role of the ``leader" whereas in~\cite{zhang2019policy}, it is required that  the player maximizing the cost be designated as the leader. As such, we clarify the algorithmic source of asymmetry in the leader-follower setup for solving this class of dynamic game problems.
  \end{enumerate}
\section{Notation and Background}
\label{sec:notations}
We use the symbols $\prec, \preceq, \succ, \succeq$ to denote the ordering induced by the positive semidefinite (p.s.d.) cone. Namely, $A \succeq B$ means that $A - B$ is positive semidefinite. For a symmetric matrix $M \in \bb R^{n \times n}$, we denote the eigenvalues in non-increasing order, i.e., $\lambda_1(M) \le \ldots \le \lambda_n(M)$.
\par
Let us recall relevant definitions and results from control theory. A matrix $A \in \bb R^{n \times n}$ is Schur if all the eigenvalues of $A$ are \emph{inside} the \emph{open unit disk} of $\bb C$, i.e., $\rho(A) < 1$ where $\rho(\cdot)$ denotes the spectral radius. A pair $(A, B)$ with $A \in \bb R^{n \times n}$ and $B \in \bb R^{n \times m}$ is stabilizable if there exists some $K \in \bb R^{m \times n}$ such that $A-BK$ is Schur. A pair $(C, A)$ is detectable if $(A^{\top}, C^{\top})$ is stabilizable. An eigenvalue $\lambda$ of $A$ is $(C, A)$-observable if
\begin{align*}
  \rank \begin{pmatrix}
    \lambda I - A \\
    C
    \end{pmatrix} = n.
\end{align*}
A matrix $K \in \bb R^{m \times n}$ is stabilizing for system pair $(A, B)$ if $A-BK$ is Schur; it is \emph{marginally stabilizing} if $\rho(A-BK) = 1$.
For fixed $A \in \bb R^{n \times n}$ and $Q \in \bb R^{n \times n}$,
  the Lyapunov matrix is of the form
  \begin{align*}
    A^{\top} X  A + Q - X= 0.
    \end{align*}
    For a system pair $(A, B)$, $Q \in \bb R^n$ and $R \in GL_n(\bb R)$,  the discrete algebraic Riccati equation (DARE) is of the form
    \begin{align}
      \label{eq:dare}
      A^{\top} X  A  - X - A^{\top} XB (R+B^{\top}XB)^{-1} B^{\top} XA + Q = 0.
      \end{align}
    Next, we recall a result on standard linear-quadratic-regulator (LQR) control.
      \begin{theorem}
        If $Q \succeq 0$, $R \succ 0$, $(A,B)$ is stabilizable and the spectrum of $A$ on the unit disk (centered at the origin) in $\bb C$ is $(Q, A)$-observable, then there exists a unique maximal solution $X^+$ to DARE~\eqref{eq:dare}.\footnote{Where the notion of maximality is with respect to the p.s.d. ordering.}
      Moreover, the infinite-horizon LQR cost is $x_{0}^{\top} X^{+} x_{0}$ and the optimal feedback control $K_*$ is stabilizing and characterized by $K_* = (R+B^{\top}X^+ B)^{-1} B^{\top} X^+ A$.
        \end{theorem}
        In the presentation, we shall refer a solution $X_0$ to~\eqref{eq:dare} as \emph{stabilizing} if the corresponding feedback gain $K_0 = (R+B^{\top}X_0B)^{-1}B^{\top} X_0 A$ is stabilizing; a solution $X_0$ is \emph{almost stabilizing} if the corresponding gain $K_0 = (R+B^{\top} X_0 B)^{-1}B^{\top}X_0A$ is marginally stabilizing.\par
          In the sequential LQ game setup, one key difference from standard LQR is that $Q$ may be indefinite in the corresponding DARE. As such, the following generalization of the above theorem becomes particularly relevant.
          \begin{theorem}
        Suppose that $Q = Q^{\top}$, $R \succ 0$, $(A,B)$ is stabilizable and there exists a solution $X$ to DARE~\eqref{eq:dare}.
      Then there exists a maximal solution $X^+$ to DARE such that the LQR cost is given by $x_0^{\top} X^{+}x_0$. Moreover the optimal feedback control is given by $K_* = (R+B^{\top}XB)^{-1} B^{\top} X^+ A$ and the eigenvalues of $A-BK_*$ lie inside the {closed} unit disk of $\bb C$.
        \end{theorem}
        For DAREs with an indefinite $Q$, we recall a theorem concerning the existence of solutions.
        \begin{theorem}[Theorem $13.1.1$ in~\cite{Lancaster1995algebraic}]
          \label{thrm:dare_existence}
          Suppose that $(A, B)$ is stabilizable, $R=R^{\top}$ is invertible, $Q = Q^{\top}$ (no definiteness assumption), and there exists a symmetric solution $\hat{X}$ to the matrix inequality,
          \begin{align*}
            \ca R(X) \coloneqq  A^{\top} X  A + Q - X - A^{\top} X B (R+B^{\top} X B)^{-1}B^{\top} X A\succeq 0,
            \end{align*}
            with $R + B^{\top} \tilde{X} B \succeq 0$.
            Then there exists a maximal solution $X^+$ to~\eqref{eq:dare} such that $R + B^{\top} X^+ B \succ 0$. Moreover, all eigenvalues of $A-B (R+B^{\top}XB)^{-1} B^{\top} X^+A$ are inside the closed unit disk.
          \end{theorem}
          The map $\ca R_{A, B, Q, R}: \bb R^{n \times n} \to \bb R^{n \times n}$ will be referred as \emph{Riccati map} in our analysis; we will also suppress its dependency on system parameters $A, B, Q, R$. In our subsequent analysis, these system parameters will not remain constant as the corresponding feedback gains are iteratively updated.
\section{LQ Dynamic Games and some of its Analytic Properties}
In this section, we review the setup of zero-sum LQ games. In particular we discuss the modified sequential formulation of LQ games and make a few analytical observations that are of independent interest. We note that some of these observations have only become necessary in the context of sequential policy updates for LQ dynamic games.
\subsection{Zero-sum LQ Dynamic Games}
In the standard setup of LQ game, we consider a (discrete-time) linear time invariant model of the form,
\begin{align*}
  x(k+1)= A x(k) - B_1 u_1(k) - B_2 u_2(k), \qquad x(0) = x_0,
\end{align*}
where $A \in \bb R^{n \times n}$, $B_1 \in \bb R^{n \times m_1}$, $B_2 \in \bb R^{n \times m_2}$, $u_1(k)$ and $u_2(k)$ are strategies played by two players. The cost incurred for both players is the quadratic cost
\begin{align*}
  J(u_1, u_2, x_0) &= \sum_{k=0}^{\infty} \left( \langle x(k), Q x(k) \rangle + \langle u_1(k), R_1 u_1(k) \rangle - \langle u_2(k), R_2 u_2(k) \rangle \right),
\end{align*}
where $Q \in \bb S_n^{+}$, $R_1 \in \bb S_{m_1}^{++} , R_2 \in \bb S_{m_2}^{++},$ and $x_0$ is the initial condition;\footnote{The notation $S_n^{+}$ and $S_n^{++}$ designate, respectively, $n \times n$ symmetric positive semidefinite and positive definite matrices.} the underlying inner product is denoted by $\langle \cdot, \cdot \rangle$. In this setting, player one chooses its policy to minimize $J$ while player two aims to maximize it. \par
The players' strategy space that we will be particularly interested in are closed-loop static linear policies, namely, policies of the form $u_1(k) = K x(k)$ and $u_2(k) = Lx(k)$, where $K \in \bb R^{m_1 \times n}$ and $L \in \bb R^{m_2 \times n}$. Note that the cost function is guaranteed to be finite over the set of \emph{Schur stabilizing feedback gains},
\begin{align*}
  \ca S = \{(K, L) \in \bb R^{m \times n} \times \bb R^{m \times n} : \rho(A-B_1 K - B_2 L) < 1\}.
\end{align*}
Indeed, if $(K, L) \in \ca S$, with initial condition $x_0$, the cost will be given by,
\begin{align*}
  J(u_1, u_2, x_0) &= x_0^{\top} \left( \sum_{j=0}^{\infty} [(A-B_1K-B_2 L)^{\top}]^{j} (Q+K^{\top}R_1 K-L^{\top}R_2L) (A-B_1K-B_2L)^j\right)x_0 \\
&= x_0^{\top} X x_0,
\end{align*}
where $X$ solves the \emph{Lyapunov matrix equation},
\begin{align}
  \label{eq:lyapunov_matrix}
    (A-B_1K-B_2L)^{\top} X(A-B_1K-B_2L) + Q+ K^{\top}R_1K-L^{\top}R_2 L = 0.
  \end{align}
  Note that~\eqref{eq:lyapunov_matrix} has a unique solution if $(K, L) \in \ca S$. We say that a pair of strategies $(K, L)$ is \emph{admissible} if $(K, L) \in \ca S$.
\begin{remark}
  An elegant result in $\ca H_{\infty}$ control theory states that the minimax problem,
  \begin{align*}
    \inf_{u_1 \in \ell_2(\bb N)} \sup_{u_2 \in \ell_2(\bb N)} \{J(u_1, u_2) \, | \,  x_{u_1, u_2} \in \ell_2(\bb N)\},
    \end{align*}
    where $\ell_2(\bb N)$ denotes the Banach space of all square summable sequences and $x_{u_1, u_2}$ denotes the state trajectory after adopting control signals $u_1, u_2$, has a unique saddle point for all initial conditions if and only if there exists two static linear feedback gains $K_*$ and $L_*$ such $u_1(k) = K_* x(k)$ and $u_2(k)=L_* x(k)$ satisfying the saddle point condition~\cite{stoorvogel1990h}. Hence, the restriction of the optimization process to static linear policies in the LQ game setting is without loss of generality.
\end{remark}
A stabilizing \emph{Nash equilibrium} for the zero-sum game is the pair of actions $\{u_1^{*}(k), u_2^{*}(k)\} =\{K_* x(k), L_*x(k)\}$ such that,
\begin{align}
  \label{eq:nash_ineq}
  J(u_1^*(k), u_2(k)) \le J(u_1^*(k), u_2^*(k)) \le J(u_1(k), u_2^*(k)),
\end{align}
for all initial states $x_0$ and all $u_1(k), u_2(k)$ for which $(u_1(k), u_2^*(k))$ and $(u_1^*(k), u_2(k))$ are both admissible pairs. We \emph{emphasize} that it is important that $(u_1(k), u_2^*(k))$ and $(u_1^*(k), u_2(k))$ are stabilizing action pairs in the inequality~\eqref{eq:nash_ineq}. To demonstrate this delicate situation, we denote by $\ca S_{\pi_i}$ the projection of $\ca S$ onto the $i$th coordinate, i.e.,
\begin{align*}
  \ca S_{\pi_1} = \{K: \exists L \text{ such that } A-B_1 K-B_2 L \text{ is Schur}\},\\
  \ca S_{\pi_2} = \{L: \exists K \text{ such that } A-B_1 K-B_2 L \text{ is Schur}\},\\
  \end{align*}
  and $\ca S_{\hat{K}}, \ca S_{\hat{L}}$ as sets defined by,
  \begin{align*}
    \ca S_{\hat{K}} = \{L: A-B_1 \hat{K}-B_2 L \text{ is Schur}\}, \\
    \ca S_{\hat{L}} = \{K: A-B_2 \hat{L}-B_1 K \text{ is Schur}\}.
    \end{align*}
    Clearly, $\ca S_{K_*} \subset \ca S_{\pi_2}$ and $\ca S_{L_*} \subset \ca S_{\pi_1}$. This means that it is not the case that for all $K \in \ca S_{\pi_1}$ and all $L \in \ca S_{\pi_2}$, the corresponding actions $u_1(k) = K x(k)$ and $u_2(k)=Lx(k)$ yield
    \begin{align*}
        J(u_1^*(k), u_2(k)) \le J(u_1^*(k), u_2^*(k)) \le J(u_1(k), u_2^*(k)).
      \end{align*}
      This is simply due to the fact that $(\hat{K}, L_*)$ is not guaranteed to be stabilizing for all $\hat{K} \in \ca S_{\pi_1}$. \par
      Note that the cost function $J$ is a function of polices $K, L$ and initial condition $x_0$. Since we are interested in the Nash equilibrium independent of the initial conditions, naturally, we should formulate cost functions for both players to reflect this independent. Indeed, this point has been discussed in~\cite{bu2019lqr} where it has been argued that such a formulation is in general necessary for the cost functions to be well defined (see details in \S III~\cite{bu2019lqr}). The independence with respect to initial conditions can be achieved by either sampling $x_0$ from a distribution with full-rank covariance~\cite{fazel2018global}, or choosing a spanning set $\{w_1, \dots, w_n\} \subseteq \bb R^n$~\cite{bu2019lqr}, and defining the value function over $\ca S$ as,
\begin{align*}
  f(K, L) = \sum_{i=1}^n J_{w_i}(K, L),
\end{align*}
where $J_{w_i}(K, L)$ is the cost by choosing initial state $w_i$, $u_1(k) = Kx(k)$ and $u_2(k) = Lx(k)$. Note that over the set $\ca S$ the value of function $f$ admits a compact form,
\begin{align*}
  f(K, L) = \Tr( X {\bf \Sigma}),
\end{align*}
where $ {\bf \Sigma} = \sum_{i=1}^n w_i w_i^{\top}$ and $X$ is the solution to~\eqref{eq:lyapunov_matrix}.
\paragraph{Behavior of $f$ on $\ca S^{c}$:}
How the cost function $f$ behaves near the boundary $\partial \ca S$ is of paramount importance in the design of iterative algorithms for LQ games. In the standard LQR problem (corresponding to a single player case in the game setup), the cost function diverges to $+\infty$ when the feedback gain approaches the boundary of this set (see~\cite{bu2019lqr} for details). In fact, this property guarantees stability of the obtained solution via first order iterative algorithms for suitable choice of stepsize. However, the behavior of $f$ on the boundary $\partial \ca S$ could be more intricate. For example, if $(K, L) \in \partial \ca S$, i.e., $\rho(A-B_1 K-B_2 L) = 1$, then it is possible that the cost is still finite for both players. This happens when an eigenvalue of $A-B_1K-B_2 L$ on the unit disk in the complex plane is not $(Q+K^{\top} R_1 K - L^{\top} R_2 L, A-B_1K -B_2L)$-observable. To see this, we observe that for every $\omega_i$, the series
\begin{align*}
  \sum_{j=0}^{\infty} J_{\omega_i}(K, L) = \omega_i^{\top} \left(\sum_{j=0}^{\infty} ((A-B_1K-B_2L)^\top)^{j} (Q+K^{\top}R_1 K-L^{\top}R_2 L)(A-B_1K-B_2L)^{j} \right)\omega_i
\end{align*}
is convergent to a finite (real) number if the marginally stable modes are not detectable.
 Even on $\bar{\ca S}^c$ (complement of closure of $\ca S$), $f$ could be finite if all the non-stable modes of $(A-B_2K-B_2L)$ are not $(Q+K^{\top} R_1 K - L^{\top} R_2 L, A-B_1K -B_2L)$-observable. The complication suggests that the function value is no longer a valid indictor of stability. We remark that such a situation does not occur in the LQ setting examined in~\cite{fazel2018global, bu2019lqr}, as it has been assumed that $Q$ is positive definite.
      \subsection{Stabilizing Policies in Sequential Zero-Sum LQ Games}
      Another subtle situation arising in sequential zero-sum LQ dynamic game is as follows: there is clearly no incentive for player $1$ to  destabilize the dynamics. However, from the perspective of player $2$, making the states diverge to infinity could be desirable as the player aims to maximize the cost. For player $1$, in the situation where $Q-L^{\top} R_2 L$ is not positive semidefinite, it is also possible that the best policy is not the one in $\ca S_{\pi_1}$. Hence, in round $j$, in order to guarantee that the game can be continued, it is important that both players choose their respective policies in $\ca S_{\pi_1}$ and $\ca S_{\pi_2}$. We may then  stipulate that both players play \emph{Schur stable} policies. We can justify this constraint by insisting that both players have an incentive to stabilize the system in the first place. This can also be encoded in the cost function for the player. That is, we may define the cost functions for player $1$ and player $2$ by,
\begin{align*}
  f_1(K, L) = \delta_{\ca S_{\pi_1}}(K) + f(K, L),\\
  f_2(K, L) = -\delta_{\ca S_{\pi_2}}(L) + f(K, L),
  \end{align*}
  where $\delta_{\ca S_{\pi_i}}(x)$ is the indicator function of the set
  \begin{align*}
    \delta_{\ca S_{\pi_i}}(x) = \begin{cases}
      0, & x \in \ca S_{\pi_i},\\
      +\infty, & x \notin \ca S_{\pi_i}.
      \end{cases}
    \end{align*}
    Then we have two cost functions defined everywhere for both players and assume a finite value on $\ca S$ which agree with each other, i.e., $f(K, L)$. We still need to be careful in realizing that there are points for which the function value is indeterminate. For example, it is possible to find a point $(\hat{K}, \hat{L})$ such that $f(\hat{K}, \hat{L}) = -\infty$; then $f_1(\hat{K}, \hat{L}) = +\infty - \infty$. To resolve this complication, we shall declare the function value to be the first summand; namely, if $f_1(\hat{K}, \hat{L}) = + \infty - \infty$, then $f_1(\hat{K}, \hat{L}) \equiv +\infty$. \par
    From the perspective of sequential algorithm design, these newly introduced cost functions would constrain both players to play policies in $\ca S$. It might be tempting to design projection based algorithms. However, this can be difficult since describing the sets $\ca S_{\pi_1}$ and $\ca S_{\pi_2}$ for given system $(A, B)$ is not straightforward. We shall see later that by exploiting the problem structure, we can design sequential algorithms for both players to guarantee this condition without any projection step.
\subsection{Analytic Properties of the Cost Function}
In this section we shall clarify analytical properties of the cost functions in terms of the polices played by the two players; that is, we consider the cost function $f(K, L)$ over $\ca S$.\footnote{In our formulation, the two players have different cost functions. But over the set $\ca S$, the cost functions coincide.} 
To begin with, we observe the set $\ca S$ even though is not convex, still possesses nice topological properties.
\begin{proposition}
  \label{prop:domain_topo}
  The set $\ca S$ is open, contractible (i.e., path-connected and simply connected) and in general non-convex.
\end{proposition}
\begin{proof}
  It suffices to note that by Kalman Decomposition~\cite{wonham2012linear}, there exists some $T \in GL_n(\bb R)$ such that,
  \begin{align*}
    T A T^{-1} = \begin{pmatrix}
      \tilde{A}_{11} & \tilde{A}_{12}\\
      0 & \tilde{A}_{22}
      \end{pmatrix}, T[B_1, B_2] = \begin{pmatrix}
        \tilde{B}_1 \\
        0
        \end{pmatrix},
    \end{align*}
    where $(\tilde{A}_{11},\tilde{B}_1)$ is controllable and $\tilde{A}_{22}$ is Schur. Suppose that $\tilde{B}_{1} \in \bb R^{n_1 \times (l_1 + l_2)}$ and further observe that $\ca S$ can be diffeomorphically identified by $\ca S_{(\tilde{A}_{11}, \tilde{B}_1)} \times \bb R^{(n-n_1) \times (l_1 + l_2)}$. The statement then follows by the results reported in~\cite{bu2019topological}.
\end{proof}
As the set $\ca S$ is generally not convex, Proposition~\ref{prop:domain_topo} assures us that algorithms based on local search (\emph{e.g.} gradient descent) can potentially reach the Nash equilibrium. If $\ca S$ had more than one path-connected components, it will be impossible to guarantee the convergence to Nash equilibrium under random initialization. Moreover this observation implies that $f$ is not convex-concave as the domain is not even convex.\par
We next observe that the value function is smooth and indeed real analytic, i.e., $f \in C^{\omega}(\ca S)$.
\begin{proposition}
 One has $f \in C^{\omega}(\ca S)$.
\end{proposition}
\begin{proof}
  For $(K, L) \in \ca S$, $f$ is the composition
  \begin{align*}
    (K, L) \mapsto X(K, L) \mapsto \Tr(X {\bf \Sigma}),
    \end{align*}
    where $X$ solves~\eqref{eq:lyapunov_matrix}. But
    \begin{align*}
      \vect(X) = \left( I \otimes I - A_{K, L}^{\top} \otimes A_{K, L}^{\top}\right)^{-1} \vect(Q+ K^{\top} R_1 K - L^{\top} R_2 L),
    \end{align*}
    by Cramer's Rule; the proof thus follows.
\end{proof}
As $f$ is smooth, its partial derivatives with respect to $K$ and $L$ can be characterized as follows.
\begin{proposition}
  On the set $\ca S$, the gradients of $f$ with respect to its arguments are given by,
  \begin{align*}
    &\partial_{K} f(K, L) = (R_1 K - B_1^{\top} X A_{K, L}) Y, \\
    &\partial_{L} f(K, L) = (-R_2 L - B_2^{\top} X A_{K, L}) Y,
    \end{align*}
    where $X$ solves the Lyapunov equation~\eqref{eq:lyapunov_matrix} and $Y$ solves the Lyapunov equation,
    \begin{align}
      \label{eq:lyapunov_matrix_Y}
      A_{K, L} Y A_{K, L}^{\top} + {\bf \Sigma} = 0.
    \end{align}
\end{proposition}
\begin{proof}
  It suffices to rewrite the Lyapunov equation in the form
  \begin{align*}
    \left( A - \begin{pmatrix}B_1 & B_2\end{pmatrix} \begin{pmatrix} K \\ L \end{pmatrix}\right)^{\top} X\left( A - \begin{pmatrix}B_1 & B_2\end{pmatrix} \begin{pmatrix} K \\ L \end{pmatrix}\right) - X + Q + \begin{pmatrix} K \\ L \end{pmatrix}^{\top} \begin{pmatrix} R_1 & 0 \\ 0 & -R_2 \end{pmatrix} \begin{pmatrix} K \\ L \end{pmatrix} = 0.
    \end{align*}
    By the result in~\cite{fazel2018global, bu2019lqr}, the gradient of $f$ is given by
    \begin{align*}
      \nabla f(K, L) = \begin{pmatrix} (R_1 K - B_1^{\top}XA_{K, L})Y \\(-R_2 K - B_1^{\top}XA_{K, L})Y \end{pmatrix}.
      \end{align*}
\end{proof}
We now observe that $Y(K,L)$ is a smooth function in $(K, L)$ and is positive definite everywhere on $\ca S$. Hence $Y(K, L)$ is a well-defined Riemannian metric on $\ca S$. Under this Riemannian metric, we can thereby identify the gradient. In learning and statistics literature, such a gradient is referred as a ``natural gradient.'' We shall use $N_{f, K}$ and $N_{f, L}$ to denote the natural gradient of $f$ over $K$ and $L$, respectively. Namely,
\begin{align*}
  N_{f, K}(K, L) = R_1 K - B_1^{\top} X A_{K, L}, \\
  N_{f, L}(K, L) = -R_2 L - B_2^{\top} X A_{K, L}.
  \end{align*}
  \subsection{A Key Assumption and its Implications}
Throughout the manuscript, we have the following standing assumption.
\begin{assumption}{1}
  \label{assump1}
  There exists a \emph{stabilizing Nash Equilibrium} $(K_*, L_*) \in \ca S$ for the zero-sum game over the system dynamic $(A,[B_1, B_2])$. Moreover, the corresponding value matrix $X_* = X_*(K_*, L_*)$ satisfies at least one of the following conditions:
\begin{itemize}
  \item[(a1):]
$R_1 + B_1^{\top} X_* B_1 \succ 0$ and $R_2 - B_2^{\top}X_* B_2 + B_2^{\top} X_* (R_1 + B_1^{\top}X_* B_1)^{-1}B_1 B_2  \succ 0$.
\item[(a2):]
  $-R_2 + B_2^{\top}X_* B_1 \prec 0$ and $R_1 + B_1^{\top} X_* B_1 - B_1^{\top}X_* B_2 (-R_2 + B_2^{\top}X_* B_2)^{-1} B_2^{\top}X_* B_1 \succ 0$.
  \end{itemize}
\end{assumption}
\begin{remark}
  The existence of a stabilizing Nash is a necessary assumption adopted in the LQ literature~\cite{basar2008h}. However, we do not constrain the value matrix $X_*$ to be positive semidefinite, as assumed in~\cite{stoorvogel1994discrete, basar2008h, zhang2019policy}. The definiteness is useful when the LQ game formulation is tied to $\ca H_{\infty}$ control. However, from the LQ game perspective, this association seems unnecessary. Conditions $(a1)$ or $(a2)$ are necessary if it is desired to extract \emph{unique} policies from the optimal value matrix $X_*$. Namely, if the \emph{total derivative} of $f$ vanishes, i.e.,
  \begin{align*}
    \begin{pmatrix} R_1 & 0 \\ 0 & -R_2 \end{pmatrix} \begin{pmatrix} K \\ L \end{pmatrix} - \begin{pmatrix} B_1^{\top} \\ B_2^{\top} \end{pmatrix} X A + \begin{pmatrix} B_1^{\top} \\ B_2^{\top} \end{pmatrix} X \begin{pmatrix} B_1 & B_2 \end{pmatrix} = 0,
    \end{align*}
    conditions $(a1)$ or $(a2)$ are sufficient to guarantee the uniqueness of the solution in $\ca S$. Indeed, assumptions $(a1)$ or $(a2)$ are ``almost necessary.'' If $(K_*, L_*)$ is a NE, then $f(\cdot, L_*)$ achieves a local minimum at $K_*$, i.e., $\nabla_{KK} f(K_*, L_*)[E, E] = \langle E, (R_1 + B_1^{\top}X_* B_1)EY_* \rangle \ge 0$ (note that by assumption, $K_*$ is in the interior of $\ca S$ and thus the second-order partial derivative is well-defined). Similarly, $\nabla_{LL} f(K_*, L_*) = -R_2 + B_2^{\top}X_* B_2 \preceq 0$. We relax these two necessary conditions to hold as strictly positive (respectively, negative) definite.\footnote{If we do not relax the semidefiniteness conditions, the NE would be solutions to GARE involving Moore-Penrose inverse. This will introduce other  complications than practically needed.} In fact, in the sequential LQ formulation, the inequalities in these two conditions correspond to certain ``quasi-Newton'' directions and as such play a central role in our convergence analysis (see \S\ref{sec:ng} and \S\ref{sec:qn} for details.).\par
    Moreover, we shall subsequently see that assumptions $(a1)$ and $(a2)$ lead to distinct choices of leaders in the sequential algorithms. More specifically, if we assume condition $(a1)$, the leader of the sequential algorithm should be player $L$; for assumption $(a2)$, player $K$ should be the designated leader.
\end{remark}
We observe several implications of this assumption.
\begin{proposition}
  Under Assumption $1$, we have following implications:
\begin{enumerate}
  \item The pair $(A, [B_1, B_2])$ is stabilizable.
    \item $X_*$ is symmetric and solves the Generalized Algebraic Riccati Equation (GARE),
\begin{align}
  \label{eq:gare}
  A^{\top} X A - X +  Q + \begin{pmatrix} B_1^{\top} X A \\ B_2^{\top}X A \end{pmatrix}^{\top} \begin{pmatrix} R_1 + B_1^{\top}X B_1 & B_1^{\top}XB_2 \\ B_2^{\top}X B_1 & -R_2 + B_2^{\top}XB_2 \end{pmatrix}^{-1} \begin{pmatrix} B_1^{\top} X A \\ B_2^{\top}X A \end{pmatrix}= 0.
  \end{align}
  \item $X_*$ is unique among all \emph{almost stabilizing solutions} of~\eqref{eq:gare}.
  \end{enumerate}
\end{proposition}
\begin{proof}
  The statement in $(a)$ is immediate since $A-B_1K_*-B_2L_*$ is Schur.\\
 In order to show $(b)$, we note that since $(K_*, L_*)$ is a stabilizing Nash Equilibrium, then $X_*$ is the solution of the Lyapunov  equation~\eqref{eq:lyapunov_matrix}; it thus follows that $X_*$ is symmetric. Further, note that the partial gradients of $f$ vanish at $(K_*, L_*)$, namely $(K_*, L_*) \in \ca S$ solves the equations
\begin{align*}
    &\nabla_{K} f(K, L) = (R_1 K - B_1^{\top} X A_{K, L}) Y = 0, \\
    &\nabla_{L} f(K, L) = (-R_2 L - B_2^{\top} XA_{K, L}) Y = 0.
\end{align*}
Substituting this in the Lyapunov equation~\eqref{eq:lyapunov_matrix}, it follows that $(K_*, L_*)$ solves the GARE~\eqref{eq:gare}. Note that the inverse,
 \begin{align*}
\begin{pmatrix} R_1 + B_1^{\top}X B_1 & B_1^{\top}XB_2 \\ B_2^{\top}X B_1 & -R_2 + B_2^{\top}XB_2 \end{pmatrix}^{-1}
  \end{align*}
  is well-defined at $X_*$ by the conditions $a1$ or $a2$ in the assumption.\\
  For the statement in $(c)$, by Lemma $3.1$~\cite{stoorvogel1994discrete}, $X_*$ is the unique stabilizing solution. It remains to show that there does not exist almost stabilizing solution to~\eqref{eq:gare}. Suppose there exists a pair $(K, L) \in \partial \ca S$, i.e., $\rho(A-B_1K-B_2L) = 1$, solving~\eqref{eq:gare} with solution $X$. Then taking the difference between the identity~\eqref{eq:gare} at $(K_*, L_*)$ and $(K, L)$, we have,
\begin{align*}
  A_*^{\top} (X_* - X) A_{K, L} = X_*-X.
  \end{align*}
  Since $A_{K, L}$ is marginally stable and $A_*$ is stable, then $I \otimes I - A_{K, L}^{\top} \otimes A_*$ is invertible and thus $X_* - X = 0$. 
  \end{proof}

  %
  \section{Oracle Models for Sequential LQ Games}
    
  In this work, we assume that both players have access to oracles that return either gradient, natural gradient or quasi-Newton directions. Suppose that $\ca O_K$ and $\ca O_L$ are the orcales for the two players respectively. The players will query their respective oracles in a sequential manner: if player $1$ query the oracle, we assume the policy played by player $2$ is fixed during the query and this policy is transparent to oracle $\ca O_K$ of player $1$. As $f$ is in general not convex-concave, if two players have the same oracles and play greedily using the information they have acquired, theoretically, there is no guarantee that they will eventually converge to the Nash equilibrium. In order to obtain theoretical guarantees, we assume that player $1$ can access an oracle that computes the minimizer of $f(K, L)$ over $K$ for a fixed $L$. This oracle can be constructed out of the simple first-order oracles by repeatedly performing gradient descent/natural gradient descent/quasi-Newton type steps. More explicitly, we shall assume that for player $1$, if player $2$'s policy is $\hat{L}$, the oracle can return $K \leftarrow \argmin_{K} f(K, \hat{L})$.
\subsection{Motivation}
We shall present the motivation for equipping player $1$ with a more powerful oracle model. As finding the Nash equilibrium is equivalent to solving the saddle point of $f(K, L)$, from the perspective of player $2$, we may associate a value function independent of player $1$. Namely, we may define a function of the form,
 \begin{align*}
   g(L) = \begin{cases}
     \inf_{K \in \ca S_{L}} f(K, L), & \text{if } L \in \ca S_{\pi_2},\\
     -\infty, & \text{otherwise}.
     \end{cases}
 \end{align*}
 If $g(L)$ possesses a \emph{smoothness} property, we may consider \emph{projected gradient ascent} over the policy space. However, reflecting over $g(L)$ would reveal that $g(L)$ is not necessarily even continuous on $\ca S_{\pi_2}$. For example, if $Q-L^{\top} R_2 L \prec 0$, then $\inf_{K} f(K, L)$ could be $-\infty$. But on the other hand, by Danskin's Theorem, $g(L)$ is differentiable at $L \in \ca S_{\pi_2}$, where $f(K, L)$ admits a unique minimizer over $K$.
 \begin{lemma}
   \label{lemma:differentiable}
   Suppose that $U \subseteq \dom(g)$ is an open subset such that for every $L \in U$, 
$$\argmin_{K \in \ca S_{L}} f(K, L)$$ 
exits and is unique. Then $g(L)$ is differentiable and its gradient is,
   \begin{align*}
     \nabla g(L) = \nabla_L f(K_L, L), \text{ where } K_L = \argmin_{K} f(K, L).
    \end{align*}
   \end{lemma}
     Traditionally Danskin's theorem would require that for every $L$, the minimization of $K$ is over a common compact set. This is not the situation in our case as $\ca S_L$ is not compact nor common over $L$\footnote{Namely, it is not necessarily true $\ca S_{L_1} = \ca S_{L_2}$ if $L_1 \neq L_2$.}. The statement of Lemma~\ref{lemma:differentiable} instead follows from a variant of Danskin's Theorem in~\cite{bernhard1995theorem}. 
   \begin{proof}
     We only need to observe that $f(K, L)$ is $C^{\infty}$ in both variables and thus Fr\'{e}chet differentiable. Hence, the assumptions of Hypothesis $D2$ in~\cite{bernhard1995theorem} are satisfied. By Theorem $D2$ in~\cite{bernhard1995theorem}, $g(L)$ is directionally differentiable in every direction. As the minimizer $K_L$ is unique, the directional derivative is uniform in every direction and consequently, $g(L)$ is differentiable.
     \end{proof}
     The next issue that needs to be addressed is whether $\ca U$ is empty. It turns out that by standard LQR theory,  $\{L \in \dom(g): Q-L^{\top} R_2 L \succ 0\}$ is a subset in $U$.
    We can thus outline an update rule assuming that $g(L)$ is Lipschitz, namely, a
    \emph{projected gradient ascent} over $L$ as,
     \begin{align*}
       L_{j+1} = P_{\ca S_{\pi_2}} \left(L_j + \eta_j \nabla g(L_j)\right),
       \end{align*}
       where $\nabla g(L_j)$ is given by
       \begin{align*}
         \nabla g(L_j) = (-R_2 L - B_2^TX_{K_{L_j}, L_j} A_{K_{L_j}, L_j}) Y_{K_{L_j}, L_j} \text{  with  } K_{L_j} = \argmin_{K} f(K, L_j).
         \end{align*}
         As already noted, we do not have a full description of the nonconvex set $\ca S_{\pi_2}$ and a projection would rather be prohibitive. What we shall propose instead, are update rules that guarantee all of its iterates stay in the set $\ca S_{\pi_2}$; this will be achieved without a projection step by exploiting the problem structure.\par
         Another interesting interpretation of our setup is to consider this game from the perspective of player $2$: we have a game played by player $2$ with a greedy adversary. Each time player $2$ chooses a policy $L'$, the adversary (player $1$) would try to act greedily, i.e., minimize the cost $f(K, L')$ over $K$. The goal for player $2$ is to achieve the Nash equilibrium point for himself/herself and the greedy adversary. The information player $2$ can acquire from the game (i.e., oracle) is the first-order information (function value and gradient). As such, player $2$ has an obligation to guarantee that along the iterates $\{L_j\}$, the oracle could return meaningful first-order information of $g(L)$, i.e., $g(L_j)$ is differentiable for every $j$.
\section{Algorithm: Natural Gradient Policy on $L$}
\label{sec:ng} 
Throughout \S\ref{sec:ng} and \S\ref{sec:qn}, we assume that condition $(a1)$ in our assumption holds, i.e.,
\begin{align*}
  R_1 + B_1^{\top}X_* B_1 \succ 0, -R_2 + B_2^{\top}X_* B_2 + B_2^{\top} X_* B_1(R_1 + B_1^{\top}X_* B_1)^{-1}B_1^{\top}X_* B_2 \prec 0,
  \end{align*}
and an oracle $\ca O_K$ that returns the stabilizing minimizer $f(K, L)$ for any fixed $L$, provided that such minimizer exists. Note that the unique minimizer corresponds to the maximal solution $X^+$ to the algebraic Riccati equation (with fixed $L$), namely,
\begin{align*}
  (A-B_2L)^{\top} X (A-B_2 L) - X + Q-L^{\top} R_2 L - (A-B_2 L)^{\top} X B_1 (R+B_1^{\top} X B_1)^{-1} B_1^{\top} X (A-B_2 L) = 0.
  \end{align*}
  We shall subsequently discuss how to construct this oracle by policy gradient based algorithms in \S\ref{sec:pg_K}.
         \subsection{Algorithm}
         The algorithm is given by:
          \begin{algorithm}[H]
            \caption{Natural Gradient Policy for LQ Game}
            \label{alg1}
            \begin{algorithmic}[1]
              \State Initialize $L_0$ such that $(A-B_1 L_0, B_2)$ is stabilizable and the DARE
\begin{align*}
  (A-B_2L_0)^{\top}X(A-B_2L_0) - X + Q - L_0^{\top}R_2 L_0 - (A-B_2 L_0)^{\top} X B_1(R_1 + B_1^{\top} X B_1)B_1^{\top}X(A-B_2L_0) = 0
  \end{align*}
  has a stabilizing solution $X^+$ with $R_1 + B_1^{\top} X^{+} B_1 \succ 0$.
              \If{ $j \ge 1$}
              \State Set: $K_{j-1} \leftarrow \argmin_K f(K, L_{j-1})$.
                  \State Set: $L_{j} = L_{j-1} + \eta_j N_g(L_j) \equiv L_{j-1} + \eta_j N_{f, L} (K_{j-1}, L_{j-1})$.
                \EndIf
              \end{algorithmic}
            \end{algorithm}
              We note that the initialization step is generally nontrivial. However, if we further assume that $(A, B_1)$ is stabilizable, $Q \succeq 0$ and those eigenvalues of $A$ lying on the unit disk are $(Q, A)$-detectable, then we can choose $L_0 = 0$\footnote{This is indeed the standard assumption in the LQ literature.}. For the general case, we may need to check invariant subspaces of system parameters (see~\cite{Lancaster1995algebraic} for details.)
            \subsection{Convergence Analysis}
              To simplify the notation let,
            \begin{align*}
              {\bf U}_{K, L} = R_1K - B^{\top} X_{K, L} A_{K, L}, \qquad {\bf V}_{K, L} = -R_2 L - B^{\top} X_{K, L}A_{K, L};
              \end{align*}
              namely $2{\bf U}_{K, L} = N_{f, K}(K, L)$ and $2{\bf V}_{K, L} = N_{f, L}(K, L)$. 
            First a useful observation.
            \begin{lemma}[NG Comparison Lemma]
              \label{lemma:ng_comparison}
              Suppose that $(K, L)$ and $(\hat{K}, \hat{L})$ are both stabilizing and let $X$ and $\hat{X}$ be the corresponding value matrices. Then
              \begin{enumerate}
                \item
                \begin{align*}
                  X-\hat{X} = A_{\hat{K}, \hat{L}}^{\top} (X - \hat{X}) A_{\hat{K}, \hat{L}} + (K - \hat{K})^{\top} {\bf U}_{K, L} + {\bf U}^{\top}_{K, L} (K- \hat{K}) - (K-\hat{K})^{\top} R_1 (K-\hat{K}) \\
                  + (L - \hat{L})^{\top} {\bf V}_{K, L} + {\bf V}_{K, L}^{\top} (L-\hat{L}) + (L-\hat{L})^{\top} R_2 (L-\hat{L}) - (A_{K, L}-A_{\hat{K}, \hat{L}})^{\top}X(A_{K, L}- A_{\hat{K}, \hat{L}}).
                \end{align*}
              \item
                \begin{align*}
                  X-\hat{X} = A_{K, {L}}^{\top} (X - \hat{X}) A_{{K}, {L}} + (K - \hat{K})^{\top} {\bf U}_{\hat{K}, \hat{L}} + {\bf U}^{\top}_{\hat{K}, \hat{L}} (K- \hat{K}) + (K-\hat{K})^{\top} R_1 (K-\hat{K}) \\
                  + (L - \hat{L})^{\top} {\bf V}_{\hat{K}, \hat{L}} + {\bf V}_{\hat{K}, \hat{L}}^{\top} (L-\hat{L}) - (L-\hat{L})^{\top} R_2 (L-\hat{L}) + (A_{K,L}-A_{\hat{K}, \hat{L}})^{\top} \hat{X}(A_{K, L}-A_{\hat{K}, \hat{L}}).
                \end{align*}
                \end{enumerate}
            \end{lemma}
            \begin{remark}
              Item $(b)$ of this lemma was observed in~\cite{zhang2019policy}. Our presentation offers a control theoretic perspective on its proof.
              \end{remark}
            \begin{proof}
              We prove item $(a)$; item $(b)$ can be proved in a similar manner. \\
              It suffices to take the difference of the Lyapunov equations:
\begin{align*}
  &A_{K, L}^{\top} X A_{K, L} + Q + K^{\top} R_1 K - L^{\top} R_2 L = X, \\
    & A_{\hat{K}, \hat{L}}^{\top} \hat{X} A_{\hat{K}, \hat{L}} + Q + \hat{K}^{\top} R_1 \hat{K} - \hat{L}^{\top} R_2 \hat{L} = \hat{X}.
  \end{align*}
  Indeed, a few algebraic operations reveal that
\begin{align*}
  X - \hat{X} &= A_{\hat{K}, \hat{L}}^{\top} (X - \hat{X}) A_{\hat{K}, \hat{L}} + (K - \hat{K})^{\top} {\bf U}_{K, L} + {\bf U}^{\top}_{K, L} (K- \hat{K}) 
- (K-\hat{K})^{\top} (R_1) (K-\hat{K}) \\
                  &\quad + (L - \hat{L})^{\top} {\bf V}_{K, L} + {\bf V}_{K, L}^{\top} (L-\hat{L}) + (L-\hat{L})^{\top} R_2 (L-\hat{L})- (A_{K, L}- A_{\hat{K}, \hat{L}})^{\top} X (A_{K, L}- A_{\hat{K}, \hat{L}}) .
  \end{align*}
            \end{proof}
            We now observe another version of comparison lemma when $L, \tilde{L} \in \dom(g)$. Indeed, this lemma will play a more prominent role in our convergence analysis.
            \begin{lemma}[Comparison Lemma $2$]
              \label{lemma:ng_comparison_2}
              Suppose that $L, \tilde{L} \in \dom(g)$, namely there exists $K, \tilde{K}$ such that
              \begin{align*}
                K = \argmin_{K' \in \ca S_{L}} f(K', L), \qquad \tilde{K} = \argmin_{K' \in \ca S_{\tilde{L}}} f(K', \tilde{L}).
              \end{align*}
              Further, suppose that the algebraic Riccati map $\ca R_{A-B_2 L, B_1, Q-L^{\top}R_2L, R_1}(\tilde{X})$ is well-defined, i.e., $R_1 + B_1^{\top} \tilde{X} B_1$ is invertible. Recall that the Riccati map is given by,
              \begin{align*}
                \ca R_{A-B_2L, B_1, Q-L^{\top}R_2 L, R_1}(\tilde{X}) &= (A-B_2L)^{\top} \tilde{X} (A-B_2L) - \tilde{X} + Q-L^{\top}R_2 L \\
                &\quad - (A-B_2L)^{\top} \tilde{X} B_1 (R_1 + B_1^{\top} \tilde{X} B_1)^{-1} B_1^{\top} \tilde{X} (A-B_2L).
              \end{align*}
              Let $X$ and $\tilde{X}$ be the corresponding value matrix. Putting
              \begin{align*}
              {\bf E} \coloneqq R_1 + B_1^{\top} \tilde{X} B_1, \qquad {\bf F} \coloneqq B_1^{\top} \tilde{X} (A-B_2 L), 
                \end{align*}
                then
                  \begin{align*}
                    X - \tilde{X} = A_{K, L}^{\top}(X-\tilde{X}) A_{K, L} + \ca R_{A-B_2 L, B_1, Q-L^{\top}R_2 L, R_1}(\tilde{X}) + ({\bf E} K - {\bf F})^{\top} {\bf E}^{-1}({\bf E}K - {\bf F}).
                    \end{align*}
                    Moreover,
                    \begin{align*}
                      \ca R_{A-B_2 L, B_1, Q-L^{\top}R_2 L, R_1} = (L - \tilde{L})^{\top} {\bf V}_{\tilde{K}, \tilde{L}} + {\bf V}_{\tilde{K}, \tilde{L}}^{\top}(L-\tilde{L}) - (L-\tilde{L})^{\top} {\bf O}_{\tilde{X}} (L-\tilde{L}),
                    \end{align*}
                    where
                    \begin{align*}
                      {\bf O}_{\tilde{X}} \coloneqq R_2-B_2^{\top} \tilde{X} B_2 +B_2^{\top} \tilde{X} B_1(R_1 + B_1^{\top} \tilde{X} B_1)^{-1} B_1^{\top} \tilde{X} B_2.
                    \end{align*}
              \end{lemma}
              \begin{proof}
                Note that $X$ solves the Lyapunov matrix equation,
                \begin{align*}
                  X = (A-B_1K-B_2L)^{\top} X (A-B_1K-B_2L) + Q - L^{\top} R_2 L + K^{\top} R_1 K,
                  \end{align*}
                  with $K = (R_1 + B_1^{\top} X B_1)^{-1} B_1^{\top} X (A-B_2 L)$. Then
                  \begin{align*}
                    X - \tilde{X} - A_{K, L}^{\top}(X-\tilde{X}) A_{K, L} &= A_{K, L}^{\top} \tilde{X} A_{K, L} - \tilde{X} + Q - L^{\top} R_2 L + K^{\top} R_1 K \\
                                                                          &= (A-B_2 L)^{\top} \tilde{X} (A-B_2 L) -\tilde{X} + Q -L^{\top} R_2 L + K^{\top}(R_1 + B_1^{\top} \tilde{X} B_1)K \\
                    &\quad - K^{\top}B_1^{\top} \tilde{X} (A-B_2L) - (A-B_2L)^{\top}\tilde{X}B_1 K \\
                                                                          &= \ca R_{A-B_2L, B_1, Q-L^{\top}R_2L, R_1}(\tilde{X})\\
                    &\quad + (A-B_2L)^{\top}\tilde{X}B_1 (R_1 + B_1^{\top} \tilde{X}B_1)^{-1}B_1^{\top}\tilde{X}(A-B_2L) \\
                    &\quad + K^{\top}(R_1 + B_1^{\top} \tilde{X} B_1)K  - K^{\top}B_1^{\top} \tilde{X} (A-B_2L) - (A-B_2L)^{\top}\tilde{X}B_1 K \\
                    &= \ca R_{A-B_2L, B_1, Q-L^{\top}R_2L, R_1}(\tilde{X}) + ({\bf E} K - {\bf F})^{\top} {\bf E}^{-1}({\bf E}K - {\bf F}).
                    \end{align*}
                Since $\tilde{X}$ satisfies the algebraic Riccati equation
                \begin{align*}
                  (A-B_2 \tilde{L})^{\top} \tilde{X} (A-B_2 \tilde{L}) + Q - \tilde{L}^{\top}R_2 \tilde{L} - (A-B_2 \tilde{L})^{\top}\tilde{X}B_1 (R_1 + B_1^{\top} \tilde{X} B_1)^{-1} B_1^{\top} \tilde{X}(A-B_2\tilde{L}) = \tilde{X},
                  \end{align*}
                  it follows that,
                  \begin{align*}
                    \ca R_{A-B_2 L, B_1, Q-L^{\top}R_2 L, R_1}(\tilde{X}) &= (A-B_2L)^{\top} \tilde{X} (A-B_2L) -L^{\top}R_2 L \\
                    &\quad - (A-B_2L)^{\top} \tilde{X} B_1 (R_1 + B_1^{\top} \tilde{X} B_1)^{-1} B_1^{\top} \tilde{X} (A-B_2L) \\
                                                               &\quad - (A-B_2 \tilde{L})^{\top} \tilde{X} (A-B_2 \tilde{L}) + \tilde{L}^{\top}R_2 \tilde{L} \\
                    &\quad + (A-B_2 \tilde{L})^{\top}\tilde{X}B_1 (R_1 + B_1^{\top} \tilde{X} B_1)^{-1} B_1^{\top} \tilde{X}(A-B_2\tilde{L}) \\
                                                               &= (L-\tilde{L})^{\top}(-R_2 \tilde{L} - B_2^{\top} \tilde{X} A_{\tilde{K}, \tilde{L}}) + (-R_2 \tilde{L}-B_2^{\top}\tilde{X}A_{\tilde{X}, \tilde{L}})^{\top}(L-\tilde{L}) \\
                    & \quad - (L-\tilde{L})^{\top}(R_2 - B_2^{\top} \tilde{X} B_2 + B_2^{\top} \tilde{X} B_1(R_1+B_1^{\top}\tilde{X} B_1)^{-1}B_1^{\top}\tilde{X} B_2) (L-\tilde{L}).
                    \end{align*}
                \end{proof}
    Let us now prove the convergence of the proposed algorithm. In the following analysis, $K_j$'s are exclusively used as the unique \emph{stabilizing minimizers},\footnote{Depending on the structure of our problem, it is possible that there exists non-stabilizing minimizers. But here we are only concerned with minimizers in the set $\ca S$.} for $f(K, L_j)$ over $K$.\footnote{Of course, we need to guarantee that for $L_j$, there exists a unique minimizer.} To simplify the notation, let
            \begin{align*}
              \Delta &\coloneqq X_{K_{j-1}, L_{j-1}}, \\
              {\bf O}_{j-1} &\coloneqq R_2-B_2^{\top}X_{K_{j-1}, L_{j-1}}B_2 +B_2^{\top}X_{K_{j-1}, L_{j-1}} B_1(R_1 + B_1^{\top} X_{K_{j-1}, L_{j-1}} B_1)^{-1} B_1^{\top} X_{K_{j-1}, L_{j-1}} B_2.
              \end{align*}
We shall first show that if Algorithm~\ref{alg1} is initialized appropriately, then with stepsize
$$\eta_{j-1} <  \frac {1}{\lambda_n({\bf O}_{j-1})},$$
Algorithm~\ref{alg1} generates a sequence $\{L_j\}$ satisfying properties listed in the following lemma. 
\begin{lemma}
              \label{lemma:ng_key_lemma}
              Suppose that Algorithm~\ref{alg1} is initialized appropriately.
              With stepsize $\eta_{j-1} \le \frac{1}{\lambda_n({\bf O}_{j-1})}$, we then have,
              \begin{enumerate}
                \item $(A-B_2 L_j, B_1)$ is stabilizable for every $j \ge 1$.
                  \item ${\bf O}_j \succ 0$ for every $j \ge 1$.
                  \item For every $j \ge 1$, $f(K, L_j)$ is bounded below over $K$ and there exists a unique minimizer $K_j$, which forms a stabilizing pair $(K_j, L_j)$. Namely, the DARE
                    \begin{align}
                      \label{eq:dare_iterate_j}
                      \begin{split}
                      0 &= (A-B_2L_j)^{\top} X (A-B_2L_j) -X\\
                      & \quad + Q - L_j^{\top}R_2 L_j - (A-B_2L_j)^{\top}X B_1(R_1 + B_1^{\top}X B_1)^{-1} B_1^{\top} X (A-B_2L_j) ,
                      \end{split}
                      \end{align}
                      admits a stabilizing maximal solution $X^+$ satisfying $R + B_1^{\top} X^{+} B_1 \succ 0$.
                \item Putting $\Lambda = X_{K_j, L_j}$, ${\bf E}_j = R_1+B_1^{\top} X_{K_j, L_j} B_1$ and ${\bf F}_j = B_1^{\top} X_{K_j, L_j}(A-B_2 L_j)$ we have
     \begin{align*}
       \Lambda - \Delta &=
       A_{K_j, L_j}^{\top}(\Lambda - \Delta) A_{K_j, L_j} + {\bf V}_{K_{j-1}, L_{j-1}}^{\top} \left( 4 \eta_{j-1} I - 4 \eta_j^2 {\bf O}_{j-1} \right) {\bf V}_{K_{j-1}, L_{j-1}} \\
& \quad + ({\bf E}_{j-1} K_j - {\bf F}_{j-1} )^{-1}{\bf E}_{j-1}^{-1} ( {\bf E}_{j-1} K_j - {\bf F}_{j-1}).
       \end{align*}
\end{enumerate}
            \end{lemma}
              \begin{proof}
                It suffices to prove the lemma by induction since all items holds at $j=0$ (by initialization of the algorithm). We shall first suppose that $(A-B_2 L_j, B_1)$ is stabilizable, i.e., $(a)$ holds. Note that this property is not automatically guaranteed and we subsequently provide an analysis to carefully remove this assumption.\footnote{A side remark on our proof strategy: in linear system theory, a number of synthesis results are developed under the assumption of stabilizability of the system. We will utilize these observations here; however, to use those tools, we must assume that $(A-B_2L_j, B_1)$ is stabilizable. But this is also one of our goals to show. The reader might recognize certain circular line of reasoning here. Indeed, one of our contributions is pseudo trick devised to circumvent this issue: we first assume stabilizability and then use the results developed to arrive at a contradiction if the system had not been stabilizable.}\\
                  First note that by our assumption, $\Delta$ is the maximal stabilizing solution of the DARE,\footnote{This can be considered as an LQR problem for system $(A-B_2 L_{j-1}, B_1)$ with state cost matrix $Q-L_{j-1}^{\top}R_2 L_{j-1}$.}
                  \begin{align*}
                    0 &=(A-B_2 L_{j-1})^{\top} X (A-B_2 L_{j-1}) - X\\
                    &\quad + Q - L_{j-1}^{\top} R_2 L_{j-1} - (A-B_2 L_{j-1})^{\top} X B_1 (R_1 + B_1^{\top} X B_1)^{-1} B_1^{\top} X (A-B_2 L_{j-1}),
                    \end{align*}
                    and $K_{j-1} = (R_1 + B_1^{\top} \Delta B_1)^{-1} B_1^{\top} \Delta (A-B_2 L_{j-1})$.
                Now adopt the update rule,
                \begin{align*}
                  L_j = L_{j-1} - \eta_{j-1} N_g(L_{j-1}) = L_{j-1} - 2\eta_{j-1} {\bf V}_{K_{j-1}, L_{j-1}}.
                  \end{align*}
                  by Lemma~\ref{lemma:ng_comparison_2} it follows
                  \begin{align}
                    \label{eq:inequality_existence}
                    \begin{split}
                      \ca R_{A-B_2L_j, B_1, Q-L_j^{\top}R_2 L_j, R_1}(X_{j-1}) = {\bf V}_{K_{j-1}, L_{j-1}}^{\top} \left( 4 \eta_{j-1} I - 4 \eta_{j-1}^2 {\bf O}_{j-1} ) \right) {\bf V}_{K_{j-1}, L_{j-1}}.
                      \end{split}
                    \end{align}
                    Thereby with the stepsize
                    $\eta_{j-1} \le \frac{1}{\lambda_n({\bf O})}$, we have
                    \begin{align}
                        \label{eq:existence_ng}
                      \ca R_{A-B_2L_j, B_1, Q-L_j^{\top}R_2 L_j, R_1}(X_{j-1}) = {\bf V}_{K_{j-1}, L_{j-1}} \left( 4 \eta_{j-1} I - 4 \eta_{j-1}^2 {\bf O}_{j-1} \right) {\bf V}_{K_{j-1}, L_{j-1}} \succeq 0.
                      \end{align}
                      By Theorem~\ref{thrm:dare_existence} (note that the inequality~\eqref{eq:existence_ng} is crucial for applying the theorem), there exists a maximal solution $X^+ \succeq \Delta$ to the DARE~\eqref{eq:dare_iterate_j} and moreover, with
                      $$K^+ = (R_1 +B_1^{\top} X^+ B_1 )^{-1} B_1^{\top} X^+(A-B_2 L_j),$$ the eigenvalues of $A-B_2 L_j - B_1 K^+$ are in the closed unit disk of $\bb C$. Equivalently, $X^+$ solves the following Lyapunov equation,
                        \begin{align*}
                          (A-B_1K^+ -B_2 L_j)^{\top} X^+ (A-B_1 K^+ - B_2 L_j) + Q + (K^{+})^{\top} R_1 K^{+} - L_j^{\top} R_2 L_j = X^+.
                          \end{align*}
                           Item $(d)$ thereby follows from the first part of Lemma~\ref{lemma:ng_comparison_2}. We now observe that $K^+$ is indeed stabilizing. Suppose that this is not the case; then there exists $v \in \bb C^n$ such that $(A-B_2 L_j - B_1 K^+) v = \lambda v$ with $|\lambda| = 1$. Hence, 
       \begin{align*}  
         &v^{\top} \left( A_{K^{+}, L_j}^{\top}(X^{+} - \Delta) A_{K^+, L_j} \right)v + v^{\top} \left({\bf V}_{K_{j-1}, L_{j-1}}^{\top} \left( 4 \eta_{j-1} I - 4 \eta_j^2 {\bf O}_{j-1} \right) {\bf V}_{K_{j-1}, L_{j-1}}  \right) v \\
         &\le  v^{\top} (X^+-\Delta) v.
         \end{align*}
         This would imply that ${\bf V}_{K_{j-1}, L_{j-1}} v = 0$. But this means that,
         \begin{align}
           \label{eq:ng_eq1}
           L_j v = L_{j-1} v.
           \end{align}
                By Lemma~\ref{lemma:ng_comparison}, we have
     \begin{align*}
       A_{K^+, L_j}^{\top}(X^+ - \Delta)  A_{K^+, L_j} - (X^+ - \Delta) + {\bf V}_{K_{j-1}, L_{j-1}}^{\top} \left( 4 \eta_{j-1} I - 4 \eta_j^2 R_2 \right) {\bf V}_{K_{j-1}, L_{j-1}} \\
       + (K^+-K_{j-1})^{\top} R_1 (K^+ - K_{j-1}) + (A_{K^+, L_j}-A_{K_{j-1}, L_{j-1}})^{\top}\Delta (A_{K^+, L_j}-A_{K_{j-1}, L_{j-1}}) = 0.
       \end{align*}
       Multiplying $v^{\top}$ and $v$ on each side and combining the resulting expression with~\eqref{eq:ng_eq1}, we obtain,
       \begin{align*}
         K_{j-1}v = K^+ v.
         \end{align*}
         But now we have
         \begin{align*}
           (A-B_1 K_{j-1} -B_2 L_{j-1})v = Av - B_1 K^{+} v - B_2 L_j v = \lambda v.
           \end{align*}
           This is a contradiction to the Schur stability of $(K_{j-1}, L_{j-1})$. For item $(b)$, note that $X_j \preceq X_*$, so $R_2 - B_2^{\top} X_j B_2 \succ 0$ and consequently ${\bf O}_j \succ 0$ as $R_1 + B_1^{\top}X_j B_1 \succeq R_1 + B_1^{\top} X_{j-1}B_1 \succ 0$. Hence, we have completed the proof for items $(b), (c), (d)$ under the assumption of item $(a)$.\par
We now argue that with the stepsize $\eta_{j-1} \le 1/\lambda_n({\bf O}_{j-1})$, this assumption of stabilizability is indeed valid; namely, item $(a)$ holds.
Consider the ray $\{L_t: L_t = L_{j-1} + t 2{\bf V}_{K_{j-1}, L_{j-1}}\}$. We first note that there exists a maximal half-open interval $[0, \sigma)$ such that $(A-B_2L_t, B_1)$ is stabilizable for every $t \in [0, \sigma)$ and $(A-B_2 L_{\sigma}, B_1)$ is not stabilizable (this is due to the fact that stabilizability is an open condition; see Proposition~\ref{prop:stabilizability_open} for a proof). Now suppose that $\sigma < \frac{1}{\lambda_n({\bf O}_{j-1})}$. We may take a sequence $t_l \uparrow \sigma$, and note that $(A-B_2 L_{t_l}, B_1)$ is stabilizable for every $t_l$. Let us denote the corresponding sequence of solutions to the DARE by $\{Z_{t_l}\}$. By our previous arguments, $\Delta \preceq Z_{t_l} \preceq X_*$, where $X_*$ is the corresponding value matrix at Nash equilibrium point $(K_*, L_*)$. Denote by $\ca L$ as the set of all limit points of the sequence $\{Z_{t_l}\}_{l=1}^{\infty}$. By Weirestrass-Balzano, the set $\ca L$ is nonempty as the sequence is bounded. Clearly, for every $Z \in \ca L$, $\Delta \preceq Z \preceq X_*$. By continuity, $Z$ must solve the DARE,
\begin{align}
  \label{eq:dare_1}
  (A-B_2L_{\sigma})^{\top} X (A-B_2 L_{\sigma}) + Q - L_{\sigma}^{\top} R_2 L_{\sigma} - (A-B_2 L_{\sigma})^{\top}X B_1 (R_1+B_1^{\top}XB_1)^{-1}B_1^{\top} X(A-B_2L_{\sigma}) = X.
  \end{align}
  Putting $K' = (R_1+B_1^{\top}ZB_1)^{-1} B_1^{\top} Z(A-B_2 L_j)$,
  we claim that $A-B_2 L_{\sigma} - B_1 K'$ is Schur stable. This is a consequence of $(d)$ and the Comparison Lemma~\ref{lemma:ng_comparison}: it suffices to observe that $A_{K', L_{\sigma}}$ is marginally stable satisfying $(d)$ and,
\begin{align*}
       A_{K', L_\sigma}^{\top}(Z - \Delta) A_{K', L_{\sigma}} - (Z-\Delta) + {\bf V}_{K_{j-1}, L_{j-1}}^{\top} \left( 4 \sigma I - 4 \sigma^2 {\bf O}_{j-1}\right) {\bf V}_{K_{j-1}, L_{j-1}} \preceq 0.
  \end{align*}
 Proceeding similar to the above line of reasoning, we can show that $A-B_2 L_\sigma - B_1 K'$ is Schur stable. But this contradicts our standing assumption that $(A-B_2 L_{\sigma}, B_1)$ is not stabilizable. Hence, for all $\eta_{j-1} \le 1/\lambda_n({\bf O}_{j-1})$, the pair $(A-B_2 L_j, B_1)$ is indeed stabilizable. 
                \end{proof}
                We are now ready to state the convergence rate for the algorithm.
                \begin{theorem}
                  If the stepsize is taken as $\eta_{j} = 1/(2\lambda_n({\bf O_{j-1}}))$, then,
                  \begin{align*}
                    \sum_{j=0}^{\infty} \|N_g(L_j)\|_F^2 \le \frac{1}{\eta} \left( g(L_*) - g(L_0)\right),
                    \end{align*}
                    where $\eta \in \bb R_+$ is some positive constant.
                  \end{theorem}
                  \begin{remark}
                    This theorem suggests the gradient will vanish at a sublinear rate. As we know, there is a unique stationary point of $g$; this means the sublinear convergence to global maximum, i.e., Nash equilibrium point.
                    \end{remark}
                  \begin{proof}
                    Let $\eta = \inf_{j} 1/\lambda_n({\bf O}_j)$ and note that ${\bf O}_j \succeq R_2 - B_2^{\top} X_* B_2$, so $\eta > 0$.
                    It suffices to note that by Lemma~\ref{lemma:ng_key_lemma}, we have
                    \begin{align*}
                      g(L_j) - g(L_{j-1}) &= \Tr((X_{K_{j}, L_j} - X_{K_{j-1}, L_{j-1}}) {\bf \Sigma}) \\
                                          &\ge \frac{1}{\lambda_n({\bf O}_{j-1})}\Tr\left[ Y_{K_j, L_j}\left(  N_g(L_{j-1})^{\top}  N_g(L_{j-1})\right)\right] \\
                                          &\ge \eta  \|N_g(L_{j-1})\|_F^2.
                     \end{align*}
                     Telescoping the sum and noting that $g(L)$ is bounded above by $g(L_*)$, we have
                     \begin{align*}
                    \sum_{j=0}^{\infty} \|N_g(L_j)\|_F^2 \le \frac{1}{\eta} \left( g(L_*) - g(L_0)\right) < \infty.
                       \end{align*}
                    \end{proof}
                    We observe that the convergence rate is asymptotically linear. This is a consequence of the local curvature of $g(L)$. Indeed, if we compute the Hessian at $g(L_*)$, the action of the Hessian (see Appendix~\ref{sec:hessian} for details) is given by
                    \begin{align*}
                        \nabla^2 g(L_*)[E, E] = -2\langle {\bf O}_{X_*} E, EY_*\rangle.
                      \end{align*}
                    As $\nabla^2 g(L_*)$ is negative definite, $-g(L)$ is locally strongly convex around a convex neighborhood of $L_*$. It thus follows that gradient descent enjoys a linear convergence rate around $L_*$. 
                    \section{Algorithm: quasi-Newton Iterations of $L$}
                    \label{sec:qn}
                    In this section, we shall assume that the oracle $\ca O_L$ returns the quasi-Newton direction. 
                    The motivation of quasi-Newton is to investigate the second-order local approximation of $g(L)$. Indeed, we may observe that,
                      \begin{align*}
                            g(L+ \Delta L) \approx g(L) + 2\langle Y_{L + \Delta L} N_g(L), \Delta L\rangle - \langle {\bf O}_L \Delta L, \Delta L \rangle.
                        \end{align*}
          \begin{algorithm}[H]
            \caption{quasi-Newton Policy for LQ Game}
            \label{alg2}
            \begin{algorithmic}[1]
              \State Initialize $(K_0, L_0) \in \ca S$ such that $(Q-L_0^{\top}R_2 L_0, A-B_2 L_0)$ is detectable and the DARE
              \begin{align*}
                (A-B_2L_0)^{\top} X (A-B_2 L_0) - X + Q - L_0^{\top}R_2 L_0 + (A-B_2L_0)^{\top}X B_1(R_1 + B_1^{\top}X B_1)B_1^{\top}X (A-B_2L_0)
                \end{align*}
                is solvable in $\ca S_{L_0} \equiv \{K \in \bb R^{m_1 \times n}: A-B_2L_0 - B_1 K \text{ is Schur}\}$.
              \If{ $j \ge 1$}
              \State Set: $K_j \leftarrow \argmin_K f(K, L_{j-1})$.
                  \State Set: $L_{j} = L_{j-1} + \eta_{j-1} {\bf O}_{j-1}^{-1}2(-R_2L_{j-1}- B_2^{\top} X_{K_{j-1}, L_{j-1}} A_{K_{j-1}, L_{j-1}})$.
                \EndIf
              \end{algorithmic}
            \end{algorithm}
            \subsection{Convergence Analysis}
            We first prove a result that can be considered as a counterpart to Lemma~\ref{lemma:ng_key_lemma}.
                        \begin{lemma}
              \label{lemma:qn_key_lemma}
              Suppose that Algorithm~\ref{alg2} is initialized appropriately.
              With stepsize $\eta_{j-1} \le \frac{1}{\lambda_n({\bf O}_{j-1})}$, we then have,
              \begin{enumerate}
                \item $(A-B_2 L_j, B_1)$ is stabilizable for every $j \ge 1$.
                  \item ${\bf O}_j \succ 0$ for every $j \ge 1$.
                  \item For every $j \ge 1$, $f(K, L_j)$ is bounded below over $K$ and there exists a unique minimizer $K_j$, which forms a stabilizing pair $(K_j, L_j)$. Namely, the DARE
                    \begin{align*}
                      (A-B_2L_j)^{\top} X (A-B_2L_j)  + Q - L_j^{\top}R_2 L_j - (A-B_2L_j)^{\top}X B_1(R_1 + B_1^{\top}X B_1)^{-1} B_1^{\top} X (A-B_2L_j) = X,
                      \end{align*}
                      admits a stabilizing maximal solution $X^+$ satisfying $R + B_1^{\top} X^{+} B_1 \succ 0$.
                \item Putting $\Lambda = X_{K_j, L_j}$, ${\bf E}_j = R_1+B_1^{\top} X_{K_j, L_j} B_1$ and ${\bf F}_j = B_1^{\top} X_{K_j, L_j}(A-B_2 L_j)$ we have
     \begin{align*}
       \Lambda - \Delta &=
       A_{K_j, L_j}^{\top}(\Lambda - \Delta) A_{K_j, L_j} + {\bf V}_{K_{j-1}, L_{j-1}}^{\top} \left( 4 \eta_{j-1} {\bf O}_{j-1}^{-1} - 4 \eta_j^2 {\bf O}_{j-1}^{-1} \right) {\bf V}_{K_{j-1}, L_{j-1}} \\
& \quad + ({\bf E}_{j-1} K_j - {\bf F}_{j-1} )^{\top}{\bf E}_{j-1}^{-1} ({\bf E}_{j-1} K_j - {\bf F}_{j-1}).
       \end{align*}
\end{enumerate}
            \end{lemma}
              \begin{proof}
                The proof proceeds similar to Lemma~\ref{lemma:ng_key_lemma}. The key difference is that the algebraic Riccati map assumes a new form with the quasi-Newton update. Namely, with quasi-Newton iteration,
                \begin{align*}
                      \ca R_{A-B_2L_j, B_1, Q-L_j^{\top}R_2 L_j, R_1}(X_{j-1}) = \left( 4 \eta_{j-1} I - 4 \eta_{j-1}^2 \right) {\bf V}_{K_{j-1}, L_{j-1}}^{\top} {\bf O}_{j-1}^{-1} {\bf V}_{K_{j-1}, L_{j-1}}.
                  \end{align*}
                  The statements then follows from almost same arguments as in Lemma~\ref{lemma:ng_key_lemma}.
  \end{proof}
            We are now ready to state the convergence rate for the algorithm.
                    \begin{theorem}
                      If the stepsize is taken as $\eta = 1/2$, then
                      \begin{align*}
                        g(L_*) - g(L_j) \le q (g(L_*) - g(L_{j-1}))^2,
                        \end{align*}
                        for some $q >0$.
                      \end{theorem}
                      \begin{proof}
                        By Lemma~\ref{lemma:qn_key_lemma}, the sequence of value matrices $\{X_j\}$ is monotonically nondecreasing and bounded above. Thereby $X_j \to X_*$ as $j \to \infty$. It follows the set $\ca E = \{X_j\} \cup \{X_*\}$ is compact. Substituting $K_{j-1} = (R_1 + B_1^{\top}X_{j-1}B_1)^{-1} B_1^{\top}X_{j-1}(A-B_2 L_{j-1})$ into the update rule, we get
                        \begin{align*}
                          L_j = - {\bf O}_{j-1}^{-1}B_2^{-1}X_{j-1} (A-B_1 (R_1+B_1^{\top}X_{j-1}B_1)^{-1}B_1^{\top}X_{j-1} A).
                          \end{align*}
                          By Lemma~\ref{lemma:ng_comparison_2} (take $\tilde{X} = X_*$ and $X = X_j$ and note ${\bf V}_{K_*, L_*} = 0$),
\begin{align*}
  X_* - X_j \preceq \sum_{\nu = 0}^{\infty} (A_j^{\top})^\nu \left( (L_*-L_j)^{\top}{\bf O}_{*} (L_*-L_j) \right)A_j^{\nu},
\end{align*}
where
\begin{align*}
  {\bf O}_* = R_2 - B_2^{\top} X_* B_2 + B_2^{\top} X_* B_1(R_1+B_1^{\top} X_* B_1)^{-1}B_1^{\top}X_*B_2.
  \end{align*}
It follows that,
\begin{align*}
  g(L_*) - g(L_j) &\le \Tr(Y_* (L_*-L_j)^{\top} {\bf O}_{*} (L_* - L_j)) \\
                  &\le \lambda_n(Y_*) \lambda_n({\bf O}_{*}) \Tr((L_* - L_j)^{\top} (L_* - L_j)).
\end{align*}
We observe
\begin{align*}
  L_* - L_j = -{\bf O}_*^{-1} B_2^{\top} X_* (A-B_1(R_1+B_1^{\top}X_* B_1)^{-1}B_1^{\top}X_* A)+ {\bf O}_{j-1}^{-1}B_2^{\top}X_{j-1} (A-B_1 (R_1+B_1^{\top}X_{j-1}B_1)^{-1}B_1^{\top}X_{j-1} A),
\end{align*}
and further, note that the map $\phi$ given by
\begin{align*}
  X \mapsto -{\bf O}_X^{-1}B_2^{\top} X \left( A-B_1 (R_1 + B_1^{\top} X B_1)^{-1} B_1^{\top} X A\right)
\end{align*}
is smooth where,
\begin{align*}
  {\bf O}_{X} =R_2 - B_2^{\top} X B_2 + B_2^{\top} X B_1(R_1+B_1^{\top}XB_1)^{-1} B_1^{\top}X B_2.
\end{align*}
So over the compact set $\ca E$, we can find a Lipschitz constant $\beta$ of $\phi$, namely, for every $X, X' \in \ca E$, we have
$\|\phi(X)-\phi(X')\|_F \le \beta \|X-X'\|_F$.
Then
\begin{align*}
  \|L_* - L_j\|_F^2 = \|\phi(X_*)-\phi(X_j)\|_F^2 \le \beta^2\|X_*-X_{j-1}\|_F^2.
\end{align*}
Hence
\begin{align*}
  g(L_*)-g(L_j) &\le c \|X_*-X_{j-1}\|_F^2 \le q \left( g(L_*)-g(L_{j-1})\right)^2,
\end{align*}
where $c, q > 0$ are constants.
                        \end{proof}
                  \section{Policy Gradient Algorithms for Solving $K \leftarrow \argmin_{K} f(K, L)$}
                  \label{sec:pg_K}
                  In this section, we describe how we can use policy gradient based oracles to solve the minimization problem for fixed $L$. If the iterates $Q-L_j{\top}R_2 L_j$ were positive definite, policy based updates, {e.g.}, gradient descent, natural gradient descent and quasi-Newton iterations for standard LQR problem as treated in~\cite{fazel2018global}~\cite{bu2019lqr}, could be adopted. Then the oracle $\ca O_K$ can be constructed by repeatedly performing the procedure to the desired precision. However, there are no guarantees that this condition would be valid in the LQ dynamic game setup. We shall prove that with the assumption that there exists a maximal stabilizing solution to the algebraic Riccati equation, gradient descent (respectively, natural gradient descent and quasi-Newton) converges to the maximal solution of the ARE at a linear (respectively, linear and quadratic) rate. In this direction, recall that the gradient, natural gradient and quasi-Newton directions for fixed $L$ are given by,
                  \begin{align*}
                    \nabla_K f(K, L) &= 2(R_1 K- B_1^{\top} X A_{K, L}) Y \eqqcolon {\bf g}(K),\\
                    N_{f, K}(K, L) &= 2(R_1 K- B_1^{\top} X A_{K, L}) \eqqcolon {\bf n}(K),\\
                    qN_{f, K}(K, L) &= 2(R_1+B_1^{\top} X B_1)^{-1}(R_1K-B_1^{\top}XA_{K, L}) \eqqcolon {\bf qn}(K).
                    \end{align*}
                  We first provide the convergence analysis for the case of natural gradient descent.
                \begin{theorem}[Natural Gradient Analysis]
                  \label{thrm:ng_inner_convergence}
                  Suppose that with fixed $L$, the ARE
                  \begin{align*}
                    (A-B_2 L)^{\top} X (A-B_2 L) - X + (A-B_2L)^{\top}XB_1(R_1 + B_1^{\top} X B_1)^{-1} B_1^{\top} X(A-B_2L) + Q-L^{\top} R_2 L = 0,
                   \end{align*}
                   has a maximal stabilizing solution $X^+$. Then the update rule,
                   \begin{align*}
                     M_{i+1} = M_i - \frac{1}{2\lambda_n(R_1 + B_1^{\top} X_i B_1)} 2(R_1 M_i - B_1^{\top} X_i (A-B_2L-B_1 M_i)),
                    \end{align*}
                    where $X_i$ solves the Lyapunov equation
                    \begin{align*}
                      (A-B_2 L - B_1 M_{i})^{\top} X_i + X_i (A-B_2 L_j - B_1 M_i) + Q - L^{\top} R_2 L + K_i^{\top} R_1 K_i = 0,
                      \end{align*}
                    converges linearly to $K^+$ provided $M_0 \in \ca S_{L}$. That is,
                    \begin{align*}
                      \|M_{i} - M_*\|_F^2 \le q^{i}\|M_0 - M_*\|_F^2,
                      \end{align*}
                      for some $q \in (0, 1)$.
                  \end{theorem}
                  \begin{remark}
                    The difference between above theorem and the standard results treated in~\cite{fazel2018global,bu2019lqr} is that we are not assuming $Q$ to be positive definite. Mind that in the above problem, the matrix $Q-L_j^{\top} R_2 L_j$ corresponding to the penalization on states in the standard LQR, can be indefinite in our sequential algorithms. The one step progression would follow from Theorem in~\cite{bu2019lqr}. However, the important difference is that the cost function is no longer coercive, requiring a separate analysis for establishing the stability of the iterates.\footnote{We note that stability of the iterates in natural gradient update was not explicitly shown in~\cite{fazel2018global}. Some of the perturbation arguments for gradient descent in~\cite{fazel2018global} can however be applied to argue for this property. Such an argument would however rely on the strict positivity of the minimum eigenvalue of $Q$.}
                    \end{remark}
                  \begin{proof}
                    The analysis in~\cite{bu2019lqr} on the one-step progression of the natural gradient descent holds here and thus the convergence rate would remain the same if we can prove that the iterates remain stabilizing. \par
                    By induction, it suffices to argue that with the chosen stepsize, $M_i$ is stabilizing provided that $M_{i-1}$ is. Consider the ray $\{M_t = M_{i-1} -  t {\bf n}(M_{i-1}): t \ge 0\}$. Note that by openness of $\ca S$ and continuity of eigenvalues, there is a maximal interval $[0, \zeta)$ such that $M_{i-1} + t {\bf n}(M_{i-1}) $ is stabilizing for $t \in [0, \zeta)$ and $M_{i-1} + \zeta {\bf n}(M_{i-1})$ is marginally stabilizing. Now suppose that $\zeta \le \frac{1}{2\lambda_n(R_1+B_1^{\top} X_{i-1} B_1)}$ and take a sequence $t_l \in [0, \zeta)$ such that $t_l \to \zeta$. Consider the sequence of value matrices $\{X_{t_l}\}$ and denote by $\ca L$ as the set of all limit points of $\{X_{t_j}\}$. Note that $\ca L$ is nonempty since the sequence is bounded by $X^{+} \preceq X_{t_l} \preceq X_{i-1}$ for every $t_l$.\footnote{Note that it is not guaranteed that $X_{t_j}$ is convergent. The limit points are also not necessarily well-ordered in the ordering induced by the p.s.d. cone.} By continuity, any $Z \in \ca L$ solves,
                    \begin{align*}
                      (A-B_1 M_{\zeta} - B_2 L)^{\top} Z  (A-B_1 M_{\zeta}-B_2 L) -Z + Q - L^{\top} R_2 L + M_{\zeta}^{\top} R_1 M_{\zeta} = 0.
                    \end{align*}
                    But this is a contradiction, since by the Comparison Lemma~\ref{lemma:ng_comparison}, we have
                    \begin{align*}
                      (A-B_1 M_{\zeta}-B_2 L)^{\top} (Z-X^{+}) (A-B_1 M_{\zeta} - B_2 L) - (Z-X^{+}) + (M_{\zeta}-K^+)^{\top} R_1(M_{\zeta}-K^+) = 0.
                      \end{align*}
                      Suppose that $(\lambda, v)$ is the eigenvalue-eigenvector pair of $A-B_1 M_{\zeta}-B_2 L_1$ such that $(A-B_1 M_{\zeta}-B_2 L_j) v = \lambda v$ and $|\lambda| = 1$. Then as $Z-X^{+} \succeq 0$, we would have $M_{\zeta} v = K^+ v$. But this is a contradiction to the assumption that $K^+$ is a stabilizing solution. Hence $\{X_i\}$ is a monotonically non-increasing sequence bounded below by $X^{+}$. As such, the sequence of iterates $\{M_i\}$ would converge linearly to $K^+$ following the arguments in~\cite{bu2019lqr}.
                    \end{proof}
                    We mention that the above stability argument can be applied for the sequence generated by the quasi-Newton iteration as well. The quadratic convergence rate for such a sequence would then follow from the proof in~\cite{bu2019lqr}.
                    \begin{theorem}[Quasi-Newton Analysis]
                  \label{thrm:qn_inner_convergence}
                  Suppose that with a fixed $L$, the ARE
                  \begin{align*}
                    (A-B_2 L)^{\top} X (A-B_2 L) - X + (A-B_2L)^{\top}XB_1(R_1 + B_1^{\top} X B_1)^{-1} B_1^{\top} X(A-B_2L) + Q-L^{\top} R_2 L = 0,
                   \end{align*}
                   has a maximal stabilizing solution $X^+$. Then the update rule
                   \begin{align*}
                     M_{i+1} = M_i - \frac{1}{2} 2(R_1 + B_1^{\top} X_{i}B_1)^{-1}(R_1 M_i - B_1^{\top} X_i (A-B_2L-B_1 M_i)),
                    \end{align*}
                    where $X_i$ solves the Lyapunov equation
                    \begin{align*}
                      (A-B_2 L_j - B_1 M_{i})^{\top} X_i + X_i (A-B_2 L_j - B_1 M_i) + Q - L_j^{\top} R_2 L_j + K_i^{\top} R_1 K_i = 0,
                      \end{align*}
                    converges quadratically to $K^+$ provided $M_0 \in \ca S$. That is,
                    \begin{align*}
                      \|M_{i} - M_*\|_F \le q\|M_0 - M_*\|_F^2,
                      \end{align*}
                      for some $q > 0$.
                  \end{theorem}
                  The gradient policy analysis requires more work since the stepsize developed in~\cite{bu2019lqr} involves the smallest eigenvalue $\lambda_1(Q)$. However by carefully replacing ``$\lambda_1(Q)$ related quantities'' in~\cite{bu2019lqr}, one can still prove the global linear convergence rate as follows.
                    \begin{theorem}[Gradient Analysis]
                  \label{thrm:gd_inner_convergence}
                  Suppose that with a fixed $L$, the ARE
                  \begin{align*}
                    (A-B_2 L)^{\top} X (A-B_2 L) - X + (A-B_2L)^{\top}XB_1(R_1 + B_1^{\top} X B_1)^{-1} B_1^{\top} X(A-B_2L) + Q-L^{\top} R_2 L = 0,
                   \end{align*}
                   has a maximal stabilizing solution $X^+$. Then the update rule
                   \begin{align*}
                     M_{i+1} = M_i - \eta_i 2(R_1 M_i - B_1^{\top} X_i (A-B_2L-B_1 M_i)) Y_{M_i, L},
                    \end{align*}
                    where $X_i$ solves the Lyapunov equation
                    \begin{align*}
                      (A-B_2 L_j - B_1 M_{i})^{\top} X_i + X_i (A-B_2 L_j - B_1 M_i) + Q - L_j^{\top} R_2 L_j + K_i^{\top} R_1 K_i = 0,
                      \end{align*}
                    converges linearly to $K^+$ provided $M_0 \in \ca S$. That is,
                    \begin{align*}
                      \|M_{i} - M_*\|_F^2 \le q^i \|M_0 - M_*\|_F^2,
                      \end{align*}
                      for some $q \in (0,1)$.
                  \end{theorem}
                  The convergence analysis of gradient policy follows closely of the idea presented in~\cite{bu2019lqr}. In~\cite{bu2019lqr}, the compactness of sublevel sets was used to devise the stepsize rule to guarantee a sufficient decrease in the cost and stability of the iterates. The proof of compactness in~\cite{bu2019lqr} however, relies on the positive definiteness of $Q-L_j^{\top} R_2 L_j$.\footnote{Or the observability of $(Q-L_j^{\top} R_2 L_j, A)$.} It is also not valid to assume that the function is coercive. But, we can show that an analogous strategy adopted in~\cite{bu2019lqr} can be employed to derive a suitable stepsize for the game setup. The details analysis are defered to Appendix~\ref{sec:gd_inner}.
                  \section{Comments on Adopting Gradient Polices for $L$}
                  Gradient policy update for LQ games has been discussed in~\cite{zhang2019policy}, where a projection step is required in updating the policy $L$. In particular, in~\cite{zhang2019policy}, it has been stated that a projection step onto the set $ \Omega = \{L: Q - L^{\top} R_2 L \succeq 0\}$ would guarantee that $L$ is stabilizing. The key issue however is the stabilizability of $(A-B_2 L, B_1)$; as such, it is not valid to assume that every $L \in \Omega$ would yield a stabilizable pair $(A-B_2 L, B_1)$. In fact, the approach adopted in our work would not work for gradient policy either, as we rely on a monotonicity property of the corresponding value matrix; gradient policy would only decrease the cost function without any guarantees to decrease the value matrix (with respect to the p.s.d. ordering).
                  If we assume that the Nash equilibrium $L_* \in \Omega$ and could guarantee that $(A-B_2L_j, B_1)$ is stabilizable, then it would be warranted that $f(K, L_j)$ has a unique minimizer for every $L_j$; in this case, the approach adopted in this paper would provide a simpler proof for convergence of gradient policies for LQ games.
                  \section{Switching the Leader in the Sequential Algorithms}
                  We shall demonstrate in this section if the condition $a2$ in the assumption holds, it might not be guaranteed that $g(L)$ is differentiable in a neighborhood of $L_*$. In this case, however, choosing the leader to be player $K$ would converge. The analysis would proceed in a similar manner. First, we observe that we can define a value function,
\begin{align*}
  h(K) = \sup_{L \in \ca S_{K}} f(K, L).
  \end{align*}
  Following virtually the same argument, we can the establish the following.
\begin{proposition}
Suppose that $\ca U \subseteq \dom(h)$ is an open set such that for every $K \in \ca U$, there is a unique maximizer of $f(K, L)$ over $L$. Then $h(L)$ is differentiable and the gradient is given by
\begin{align*}
  \nabla h(K) = \nabla_{K} f(K, L_K), \text{ where } L_K = \argmax_{L \in \ca S_K} f(K, L).
  \end{align*}
\end{proposition}
The algorithm for player $K$ using natural gradient policy can be described similarly to Algorithm~\ref{alg1}.
          \begin{algorithm}[H]
            \caption{Natural Gradient Policy for LQ Game}
            \label{alg3}
            \begin{algorithmic}[1]
              \State Initialize $K_0$ such that $(A-B_1 K_0, B_2)$ is stabilizable and the DARE 
\begin{align*}
  (A-B_1K_0) Z (A-B_1 K_0) - Q - K_0^{\top}R_1 K_0 - (A-B_1K_0)^{\top}ZB_2(R_2 + B_2^{\top}ZB_2)^{-1}B_2^{\top}Z(A-B_1 K_0) = Z.
\end{align*}
has a maximal symmetric solution $Z^+$ with $R_2 + B_2^{\top}Z^+ B_2 \succ 0$.
              \If{ $j \ge 1$}
              \State Set: $L_{j-1} \leftarrow \argmax_L f(K_{j-1}, L)$.
                  \State Set: $K_{j} = K_{j-1} - \eta_j N_h(L_j) \equiv K_{j-1} - \eta_j N_{f, K} (K_{j-1}, L_{j-1})$.
                \EndIf
              \end{algorithmic}
            \end{algorithm}
            We observe that for fixed $K'$, if $L'$ is the unique stabilizing maximizer of $f(K', L)$ over $L$, then substituting $\nabla_L f(K', L') = 0$ into the Lyapunov matrix equation, $L'$ solves the following ARE,
\begin{align}
  \label{eq:max_are}
  (A-B_1K') Z (A-B_1 K') + Q + K'^{\top}R_1 K' - (A-B_1K')^{\top}ZB_2(-R_2 + B_2^{\top}ZB_2)^{-1}B_2^{\top}Z(A-B_1 K') = Z.
\end{align}
To utilize the theory developed in standard ARE, which concerns a minimization problem, we may consider following modification:
\begin{align}
  \label{eq:max_mod_are}
  (A-B_1K') W (A-B_1 K') - Q - K'^{\top}R_1 K' - (A-B_1K')^{\top}WB_2(R_2 + B_2^{\top}WB_2)^{-1}B_2^{\top}W(A-B_1 K') = W.
\end{align}
We observe that if $W$ solves~\eqref{eq:max_mod_are}, then $-W$ solves~\eqref{eq:max_are}. Now the analysis can be done almost in parallel.
\begin{lemma}
              Suppose that $L_{j-1}$ is the unique stabilizing maximizer of $f(K_{j-1}, L)$, i.e., $L_{j-1} = \argmin_{L} f(K_{j-1}, L)$. Putting $\Delta = X_{K_{j-1}, L_{j-1}}$ and
\begin{align*}
  {\bf O}_{j-1} = R_1 + B_1^{\top} \Delta B_1 + B_1^{\top} \Delta B_2 (R_2 - B_2^{\top}\Delta B_2)^{-1} B_2^{\top} \Delta B_1,
\end{align*}
with stepsize $\eta_{j-1} \le \frac{1}{\lambda_n({\bf O}_{j-1})}$, we then have,
              \begin{enumerate}
                \item $(A-B_1 K_j, B_2)$ is stabilizable.
                \item $f(K_j, L)$ is bounded above over $L$ and there exists a unique stabilizing maximizer $L_j$, namely $(K_j, L_j)$ is a stabilizing pair.
                \item Putting $\Lambda = X_{K_j, L_j}$, we have
     \begin{align*}
       \Delta - \Lambda &\succeq
       A_{K_j, L_j}^{\top}(\Delta - \Lambda) A_{K_j, L_j} + {\bf U}_{K_{j-1}, L_{j-1}}^{\top} \left( -4 \eta_{j-1} I + 4 \eta_j^2 {\bf O}_{j-1} \right) {\bf U}_{K_{j-1}, L_{j-1}}.
       \end{align*}
\end{enumerate}
\end{lemma}
\begin{proof}
  The proof proceeds similarly to Lemma~\ref{lemma:ng_key_lemma}. Indeed, putting $\tilde{\Delta} = - \Delta$ and $\tilde{\Lambda} = -\Lambda$, we observe $\tilde{\Delta}$ and $\tilde{\Lambda}$ solves the DARE~\eqref{eq:max_mod_are} and also note the DARE~\eqref{eq:max_mod_are} has the same form with the DARE considered in Lemma~\ref{lemma:ng_key_lemma}. Further, note that the update rule is equivalent to
\begin{align*}
  K_{j} &= K_{j-1} - \eta_{j-1} 2 (R_1 K_{j-1} - B_1^{\top}\Delta A_{K_{j-1}, L_{j-1}}) = K_{j-1} - \eta_{j-1} 2(R_1 K_{j-1} + B_1^{\top} \tilde{\Delta} A_{K_{j-1}, L_{j-1}})\\
        &=K_{j-1} + 2 \eta_{j-1}(-R_1K_{j-1} - B_1^{\top}\tilde{\Delta}A_{K_{j-1}, L_{j-1}}).
  \end{align*}
  In view of these observations, by the same machinery we employed in Lemma~\ref{lemma:ng_key_lemma}, we conclude
\begin{align*}
  \tilde{\Lambda} - \tilde{\Delta} &\succeq  A_{K_j, L_j}^{\top}( \tilde{\Lambda} - \tilde{\Delta}) A_{K_j, L_j} + \tilde{{\bf U}}_{K_{j-1}, L_{j-1}}^{\top} \left( 4 \eta_{j-1} I - 4 \eta_j^2 \tilde{{\bf O}}_{j-1} \right) \tilde{{\bf U}}_{K_{j-1}, L_{j-1}},
  \end{align*}
  where
\begin{align*}
  \tilde{\bf U}_{K_{j-1}, L_{j-1}} &= -R_1 K_{j-1} - B^{\top} \tilde{\Delta} A_{K_{j-1}, L_{j-1}}, \\
  \tilde{{\bf O}}_{j-1} &= R_1 - B_1^{\top} \tilde{\Delta} B_1 + B_1^{\top} \tilde{\Delta} B_2 (R_2 + B_2^{\top}\tilde{\Delta} B_2)^{-1} B_2^{\top} \tilde{\Delta} B_1.
  \end{align*}
  It thus follows that,
\begin{align*}
  -\Lambda + \Delta &\succeq  
       A_{K_j, L_j}^{\top}(\Delta - \Lambda) A_{K_j, L_j} + {\bf U}_{K_{j-1}, L_{j-1}}^{\top} \left( -4 \eta_{j-1} I + 4 \eta_j^2 {\bf O}_{j-1} \right) {\bf U}_{K_{j-1}, L_{j-1}}.
  \end{align*}
\end{proof}
Now it is straightforward to conclude the sublinear convergence rate of Algorithm~\ref{alg3}.
\begin{lemma}
  Suppose $\{K_j\}$ are the iterates generated by Algorithm~\ref{alg3}. Then we have
  \begin{align*}
    \sum_{j=0}^{\infty} N_h(K_j) \le \eta \left( h(K_0) - h(K_*)\right),
    \end{align*}
    where $\eta > 0$ is some positive number.
\end{lemma}
The analysis of quasi-Newton method with $K$ as the leader proceeds in a similar manner as Algorithm~\ref{alg2}; as such we omit the details here.
                  \section{Concluding Remarks}
                  The papers considers sequential policy-based algorithms for LQ dynamic games. We prove global convergence of several mixed-policy algorithms as well as identifying the role of control theoretic constructs in their analysis. Moreover, we have clarified a number of intricate issues pertaining to stabilization for LQ games and indefinite cost structure, while removing restrictive assumptions and circumventing the projection step.
                   \section*{Acknowledgements}
                  The authors thank Henk van Waarde and Shahriar Talebi for many helpful discussions.
           %
\section*{Appendix}
       
                  \begin{appendix}
                    \section{Hessian of $g(L)$}
                    \label{sec:hessian}
                    In this section, we compute the Hessian of $g(L)$ at a point of differentiation of $L_0$. Indeed, we shall assume stronger assumptions of $L_0$: there is a unique stabilizing minimizer of $f(K, L_0)$ over $L_0$, denoted by $K_0$. Throughout the section, we denote $A_0 = A - B_1K_0 - B_2L_0$. Note, by assumption the DARE is solvable
                    \begin{align*}
                      (A-B_2L_0)^{\top}X(A-B_2L_0) + Q-L_0^{\top}R_2L_0 + (A-B_2L_0)^{\top}XB_1(R_1+B_1^{\top}XB_1)^{-1}B_1^{\top}X(A-B_2L_0)=X,
                      \end{align*}
                      and the maximal solution $X_0$ is stabilizing, i.e., $K_0 = (R_1+B_1^{\top}X_0B_1)^{-1}B_1^{\top}X(A-B_2L_0)$ is stabilizing the system $(A-B_2L_0, B_1)$. As we have noted, the gradient of $g(L_0)$ is given by
                      \begin{align*}
                        \nabla g(L_0) = 2(-R_2 L_0 - B_2^{\top} X_0 (A-B_2 L_0 - B_1 K_0))Y_0,
                        \end{align*}
                        where $Y_0$ is the solution to the Lyapunov matrix equation
                        \begin{align*}
                          A_0Y A_0^{\top} + {\bf \Sigma} = Y.
                          \end{align*}
                        We now compute the Fr\'{e}chet derivative of $\phi(L_0) = 2(-R_2 L_0 - B_2^{\top}X_0 A_0)$. Note $\phi: \bb R^{m_2 \times n} \to \bb R^{m_2 \times n}$, the Fr\'{e}chet derivative $D\phi(L_0)$ is a bounded linear map in $\ca L(\bb R^{m_2 \times n}, \bb R^{m_2 \times n})$. So the action of $D\phi(L_0)$ at any $E \in \bb R^{m_1 \times n}$, denoted by $D\phi(L_0)[E] \eqqcolon {\bf D} \in \bb R^{m_2 \times n}$ is given by
                        \begin{align*}
                          {\bf D} = 2(-R_2 E - B_2^{\top}X_0'(E)(A-B_2L_0 - B_1 K_0) -B_2^{\top}X_0(-B_2 E) - B_2^{\top}X_0 (-B_2 K_0'(E))),
                          \end{align*}
                          where $X_0'(E) \in \bb R^{n \times n}$ (respectively, $K_0'(E)$) is the action of the Fr\'{e}chet derivative of $X_0$ (respectively, $K_0$) with respect to $L_0$. Here we concern $X_0, K_0$ as maps of $L$. Now $X_0'(E)$ satisfies
                          \begin{align*}
                            A_{0}^{\top} X_0'(E) A_0 - X_0'(E) - E^{\top} R_2 L_0 - L_0^{\top}R_2 E + K_0'(E)^{\top}R_1 K_0 + K_0^{\top}R_1 K_0'(E)\\
                            -(B_2 E + B_1 K_0'(E))^{\top}X_0 A_0 -A_0^{\top}X_0 (B_1K_0'(E)+ B_2 E) = 0.
                            \end{align*}
                            Noting $R_1K_0 - B_1^{\top}X_0 A_0 =0$, we have $X_0'(E)$ is the solution to the following Lyapunov equation
                            \begin{align*}
                            A_{0}^{\top} X_0'(E) A_0 - X_0'(E) - E^{\top} (R_2 L_0 + B_2^{\top}X_0 A_0) - L_0^{\top}(R_2 E + A_0^{\top}X_0 B_2 E) = 0.
                              \end{align*}
                              As $A_0$ is Schur, the solution exists and is unique
                              \begin{align*}
                                X_0'(E) = \sum_{j=0}^{\infty} (A_0^{\top})^j [- E^{\top} (R_2 L_0 + B_2^{\top}X_0 A_0) - L_0^{\top}(R_2 E + A_0^{\top}X_0 B_2 E) ] A_0^j.
                                \end{align*}
                                Similarly, we may compute
                                \begin{align*}
                                  K_0'(E) &= (R_1 + B_1^{\top}X_0B_1)^{-1}B_1^{\top}X_0(-B_2E) + (R_1+B_1X_0B_1)^{-1}B_1^{\top}X_0'(E) (A-B_2L_0)\\
                                  &\quad + (R_1+B_1^{\top}X_0 B_1)^{-1} B_1^{\top}X_0'(E)B_1 (R_1 +B_1^{\top}X_0B_1)^{-1} B_1^{\top}X_0 (A-B_2L_0),
                                  \end{align*}
                                  and
                                  \begin{align*}
                                    Y_0'(E) = \sum_{j=0}^{\infty} A_0^j \left[ A_0Y_0(-B_2E)^{\top} + (-B_2E)Y_0 A_0^{\top}\right](A_0^{\top})^j.
                                    \end{align*}
                                  Combining the computations, we have the action of the Hessian is given by
                                  \begin{align*}
                                    \langle \nabla^2 g(L) E, E \rangle &= 2\langle -R_2 + B_2^{\top}X_0B_2 - B_2^{\top}X_0 B_1(R_1+B_1^{\top}X_0 B_1)^{-1}B_1^{\top}X_0B_2)E, E\rangle \\
                                                                       &\quad + 2\langle -B_2^{\top}X_0'(E)A_0-B_2^{\top}X_0 B_1 (R_1+B_1X_0B_1)^{-1}B_1^{\top}X_0'(E) (A-B_2L_0),  E\rangle \\
                                                                       &\quad + 2\langle B_2^{\top}X_0B_1(R_1+B_1^{\top}X_0 B_1)^{-1} B_1^{\top}X_0'(E) B_1 (R_1 +B_1^{\top}X_0B_1)^{-1} B_1^{\top}X_0 (A-B_2L_0),E\rangle \\
                                    &\quad + 2 \langle (-R_2 L_0 - B_2^{\top} X_0 A_0)Y_0'(E), E\rangle
                                    \end{align*}
                                  It is instructive to note that at $L_*$, $X_*'(E) = 0$ since $-R_2 L_* - B_2^{\top}X_* A_* = 0$. So the action of Hessian at $L_*$ is given by
                                  \begin{align*}
                                    \nabla^2 g(L_*)[E, E] = \langle \left( -R_2 + B_2^{\top}X_0B_2 - B_2^{\top}X_0 B_1(R_1+B_1^{\top}X_0 B_1)^{-1}B_1^{\top}X_0B_2)\right)EY_0, E\rangle.
                                    \end{align*}
                                    That is, $\nabla^2 g(L_*)$ is a positive definite operator by condition $(a1)$ in the assumption. Hence, $-g(L)$ is locally strongly convex in a convex neighborhood.
                                    \begin{remark}
                                      If we ignore the formality, the above computation is nothing but a linear approximation.
                                      \end{remark}
    \section{Gradient Policy Analysis for Nonstandard LQR}
    \label{sec:gd_inner}
    This section is devoted to the proof of Theorem~\ref{thrm:gd_inner_convergence}. As it was pointed out previously, the strategy for devising a stepsize guaranteeing linear convergence for the LQ game setup is similar to the one presented in~\cite{bu2019lqr}. However, the convergence analysis for the game setup is more involved as one can not estimate the needed quantities using the current function values due to the indefiniteness of the term $Q-L^{\top}R_2L$. However, as we will show, a perturbation bound would circumvent this issue and allows deriving the required stepsize.\par
In order to simplify the notation, let
\begin{align*}
\psi(M) \coloneqq f(M, L), \qquad {\bf U} = R_1M-B_1^{\top}X(A-B_2L-B_1M), \qquad A_M = A-B_2L-B_1 M.
\end{align*}
 Note that in our analysis $L$ is always fixed; we also adopt the notation $A_L \coloneqq A-B_2L$. \\
If $M_\eta = A-B_2 L - B_1 (M - \eta 2 {\bf U} Y)$ is stabilizing,     \begin{align*}
      \psi(M) - \psi(M_\eta) = 4\eta \Tr\left( {\bf U}^{\top} {\bf U} (Y Y(\eta) - \eta a Y Y_{\eta} Y) \right),
    \end{align*}
    where $a = \lambda_n(R_1 + B_1^{\top}M B_1)$, and $Y(\eta)$ solves the Lyapunov matrix equation
    \begin{align*}
      (A-B_2L-B_1 M_{\eta})^{\top} Y(\eta) (A-B_2L-B_1 M_{\eta}) + {\bf \Sigma} = Y(\eta).
    \end{align*}
    Now define a univariate function $\phi$ as,
    \begin{align*}
      \phi(\eta) = \Tr\left( {\bf U}^{\top} {\bf U} (Y Y(\eta) - \eta a Y Y_{\eta} Y) \right).
    \end{align*}
    We observe that $\phi(\eta)$ is well-defined locally around $0$ by openness of the set $\ca S_{L}$.\footnote{$\phi$ is well-defined only if $A-B_2 L - B_1 M_{\eta}$ is Schur.}
Further, observe that $\phi(0) > 0$ if the gradient does not vanish at $M$. Now our goal is to characterize a step size such that $\phi(\eta) > 0$. In this direction, we first observe a perturbation bound on $Y(\eta)$.
    \begin{proposition}
Putting $\mu_1 = \|Y\|_2  \|B_1 {\bf U} Y\|_2^2/\lambda_1({\bf \Sigma})$ and $\mu_2 = \|Y\|_2 \|B_1 {\bf U}Y\|_2\|A-B_2L\|_2/\lambda_1({\bf \Sigma}) $, if we let
\begin{align*}
              \eta_0 = \frac{\sqrt{1+ \frac{\mu_2^2}{\mu_1}}}{2\sqrt{\mu_1}} - \frac{\mu_2}{2\mu_1},
  \end{align*} and supposing that
      $A_{\eta} = A-B_2L-B_1 M_{\eta}$ is Schur stable for every $\eta \le \eta_0$, then for all $\eta \le \eta_0$,
      \begin{align*}
        \|Y(\eta)\|_2 \le \beta_0 \|Y\|_2,
        \end{align*}
        where $\beta_0 = \frac{1}{1- 4 \mu_1 \eta_0^2 - 4\mu_2 \eta_0} > 0$.
      \end{proposition}
      \begin{proof}
        Taking the difference of the corresponding Lyapunov matrix equations, we have
        \begin{align*}
          Y(\eta) -Y - (A_L-B_1M) (Y(\eta)-Y) (A_L-B_1M)^{\top} &= A_L Y(\eta) 2 \eta (B_1{\bf U} Y)^{\top} + 2\eta B_1{\bf U} Y Y(\eta) A_L \\
          &\quad + 4\eta^2 B_1{\bf U}Y Y(\eta) (B_1 {\bf U} Y)^{\top}\\
                                                         &\preceq \|Y_{\eta}\|_2 \left( 4\eta \|B_1 {\bf U} Y\|_2\|A_L\| + 4\eta^2 \|B_1 {\bf U}Y\|_2^2 \right) I \\
          &\preceq \|Y_{\eta}\|_2 \left( 4\eta \|B_1 {\bf U} Y\|_2\|A_L\| + 4\eta^2 \|B_1 {\bf U}Y\|_2^2 \right) \frac{{\bf \Sigma}}{\lambda_1({\bf \Sigma})}.
        \end{align*}
        It thus follows that,
        \begin{align*}
          Y_{\eta} - Y \preceq \frac{\|Y_{\eta}\|_2 \left( 4\eta \|B_1 {\bf U} Y\|_2\|A_L\| + 4\eta^2 \|B_1 {\bf U}Y\|_2^2 \right) }{\lambda_1({\bf \Sigma})} Y.
          \end{align*}
        Hence,
          \begin{align*}
            \|Y_{\eta}\|_2 \left( 1- \frac{\|Y\|_2 \left( 4\eta \|B_1 {\bf U} Y\|_2\|A_L\| + 4\eta^2 \|B_1 {\bf U}Y\|_2^2 \right) } {\lambda_1({\bf \Sigma})} \right) \le \|Y\|_2.
            \end{align*}
            The proof is completed by a direct computation and noting that $1/\beta_0 = 1-\mu_1 \eta_0^2 - 4 \mu_2 \eta_0 > 0$ with the choice of $\eta_0$ and for every $\eta \le \eta_0$, 
            \begin{align*}
              1-4\mu_1 \eta^2 - 4 \mu_2 \eta \ge 1-4\mu_1 \eta_0^2 - 4 \mu_2 \eta_0.
            \end{align*}
        \end{proof}
        We now present an important result for our analysis. The basic idea of this lemma is as follows: if $[0,c)$ is the largest interval such that $A_{t}$ is Schur stable for every $t \in [0, c)$ and $A_c$ is marginally Schur stable,\footnote{Such a $c$ exists by the openness of the set of stabilizing gains.} then we can find a number $c_0 < c$ such that $\psi(M_s) \le \psi(M)$ for every $s \in [0, c_0]$.
    \begin{lemma}
      Let $c$ be the largest real positive number such that $A_{t}$ is Schur stable for every $t \in [0, c)$ and $A_c$ is marginally Schur stable\footnote{Here we have assumed that $c$ is not $+\infty$. Of course, if $c=+\infty$, then any stepsize would remain stabilizing.}. Let
\begin{align*}
  a_1 &= a \beta_0\lambda_n(Y) + 4\|{\bf U}\|_2 \beta_0 [\lambda_n(Y)]^2,\\
  a_2 &= a 4\|{\bf U}\|_2 \beta_0 [\lambda_n(Y)]^2;
  \end{align*}
then with
      \begin{align*}
        \eta_1 \le \min(c-\varepsilon, \eta_0, c_0),
        \end{align*}
        where $\varepsilon > 0$ is an arbitrary positive real number and
        \begin{align*}
          c_0 < \sqrt{\frac{ 1}{a_2} + \frac{a_1^2}{4a_2^2}} - \frac{a_2}{2a_1},
          \end{align*}
        one has $\phi(\eta_1) \ge 0$.
      \end{lemma}
      \begin{proof}
        The computations follow a similar method used in~\cite{bu2019lqr} by replacing the estimate of $Y(\theta)$ by the bound in the above proposition (see details in Lemma $5.5$ in~\cite{bu2019lqr}). 
        \end{proof}
        If one could explicitly compute the $c$ in the above result, then a deterministic choice of stepsize could be chosen; however, this is not feasible. Fortunately, we can show that $c > \min(\eta_0, c_0)$. This would then imply that one can choose the stepsize $\eta = \min(\eta_0, c_0)$.
        \begin{theorem}
          With the stepsize $\eta = \min(\eta_0, c_0)$,  $M_{\eta}$ remains stabilizing and $\phi(\eta) \ge 0$.
          \end{theorem}
\begin{proof}
 Let $\eta = \min(\eta_0,c_0)$. It suffices to prove that for every $t \in [0, \eta]$, $A_t$ is Schur stabilizing and $\phi(t) \ge 0$. We prove this by contradiction. Suppose that this is not the case. Then by continuity of eigenvalues, there exists a number $\eta' \le \eta$ such that $A_s$ is stabilizing for every $s \in [0, \eta')$ and $M_{\eta'}$ is stabilizing. If this is the case, the choice of $\eta_0, c_0$ guarantees that for every $s \in [0, \eta')$, $\phi(s)$ is well-defined and $\phi(s) \ge 0$.
Now take a sequence $t_i \to \eta'$ and consider the corresponding sequence of value matrices $\{X_{t_i}\}$. Note that the sequence of function values $\Tr(X_{t_i} {\bf \Sigma})$ satisfies,
\begin{align*}
  \Tr(X^+{\bf \Sigma}) \le \Tr(X_{t_i}{\bf \Sigma}) \le \Tr(X{\bf \Sigma}),
  \end{align*}
  since $\phi(t) \ge 0$.
But this implies that $\{X_{t_i}\}$ is a bounded sequence (note that the above inequality on function values does not guarantee the boundedness of the sequence; it is crucial that $X_{t_i} \succeq X^+$). 
                            Hence by a similar argument adopted in the proof of Theorem~\ref{thrm:ng_inner_convergence}, these observations establish
                             a contradiction, and as such, the proposed stepsize guarantees stabilization.
  \end{proof}
  It is now straightforward to conclude the convergence rate.
  \section{A Useful Control Theoretic Observation}
  In the proof to Lemma~\ref{lemma:ng_key_lemma} and~\ref{lemma:qn_key_lemma}, we have used the fact that the set $\{L: (A-B_2L, B_1) \text{ is stabilizable}\}$ is open; here is the justification.
  \begin{proposition}
    \label{prop:stabilizability_open}
    Suppose $A \in \bb R^{n \times n}$, $B_1 \in \bb R^{n \times m_1}$ and $B_2 \in \bb R^{n \times m_2}$ are fixed. Then the set
    \begin{align*}
      \ca L = \{L \in \bb R^{m_2 \times n}: (A-B_2 L, B_1) \text{ is stabilizable}\}
      \end{align*}
      is open in $\bb R^{m_2 \times n}$.
    \end{proposition}
    \begin{proof}
      Recall a pair $(A-B_2L, B_1)$ is stabilizable if and only if there exists some $F \in \bb R^{m_1 \times n}$ such that $A-B_2 L-B_1 F$ is Schur. So $(A-B_2L, B_1)$ is stabilizable if and only if there exists $X \succ 0$ and $F \in \bb R^{m_1 \times n}$ such that
      $$(A-B_2L-B_1F)^{\top} X (A-B_2L-B_1F) - X \prec 0.$$
      Now consider the map $\psi: \bb S_n^{++}  \times \bb R^{m_2 \times n} \times \bb R^{m_1 \times n} \to \bb R$ by
      \begin{align*}
        (X, L, F) \mapsto (A-B_2L-B_1F)^{\top} X (A-B_2L-B_1F) - X \\
        \mapsto \lambda_{\max}\left( (A-B_2L-B_1F)^{\top} X (A-B_2L-B_1F) - X \right).
        \end{align*}
        The map $\psi$ is continuous as it is a composition of continuous maps. It thus follows that $\psi^{-1}( (-\infty, 0) )$ is open. We now observe that $\ca L \equiv \pi_2(\psi^{-1}(-\infty, 0))$ where $\pi_2$ is the projection onto the second coordinate. Since projection map is open map\footnote{A map $f: X \to Y$ is open if $f(U)$ is open in $Y$ whenever $U \subseteq X$ is an open set.}, $\ca L$ is open.
            \end{proof}
                    \end{appendix}
                    
\bibliographystyle{alpha}
\bibliography{ref}
\end{document}

%% file: pkgs.tex
\usepackage[utf8]{inputenc}

\usepackage{tkz-graph}

\usepackage{cancel}
\usepackage[normalem]{ulem}
\usepackage[T1]{fontenc} 
\usepackage{color}
\usepackage[leqno]{amsmath}
\usepackage{amsfonts, amsthm, amssymb, commath} 
\usepackage{mathtools}
\usepackage{graphicx}
\usepackage{mathrsfs}
\usepackage{caption}
\usepackage{subcaption}
\usepackage{listings}
\usepackage[]{eufrak}
\usepackage[]{hyperref}
\usepackage[]{float}
\usepackage{tikz}
\usetikzlibrary{calc, arrows,matrix,positioning,scopes}
\usetikzlibrary{shapes.geometric}
\usetikzlibrary{decorations.markings}
\usepackage{tikz-cd}
\usetikzlibrary{cd}
\usepackage{booktabs}
\usepackage{enumerate}
\usepackage{eufrak}
\usepackage{MnSymbol}
\usepackage{tikz}
\usepackage{algorithm,algpseudocode}
\usetikzlibrary{matrix,decorations.pathreplacing,calc}
\pgfkeys{tikz/mymatrixenv/.style={decoration=brace,every left delimiter/.style={xshift=3pt},every right delimiter/.style={xshift=-3pt}}}
\pgfkeys{tikz/mymatrix/.style={matrix of math nodes,left delimiter={(}, right delimiter={)}, inner sep=1pt,column sep=1.5em,row sep=0.5em,nodes={inner sep=0pt}}}
\pgfkeys{tikz/mymatrixbrace/.style={decorate,thick}}

\usepackage{cuted}
\usepackage{fancyhdr}

%% file: defs.tex
\newcommand{\Tr}{\mathop{\bf Tr}}
\newcommand{\ca}[1]{\mathcal{#1}}

\newcommand{\bb}[1]{\mathbb{#1}}

\newcommand{\rank}{\text{rank}}

\newcommand{\vect}{\mbox{\bf vec}}

\newcommand{\argmin}{\mathop{\rm arg\,min}}
\newcommand{\argmax}{\mathop{\rm arg\,max}}

\newcommand{\dom}{\mathop{\bf dom}} 

\makeatletter
\newcommand{\leqnomode}{\tagsleft@true}
\newcommand{\reqnomode}{\tagsleft@false}
\makeatother

\newtheorem{theorem}{Theorem}[section]

\newtheorem{lemma}[theorem]{Lemma}
\newtheorem{proposition}[theorem]{Proposition}

\newtheorem{remark}[theorem]{Remark}

\newtheorem*{assumption*}{\assumptionnumber}
\providecommand{\assumptionnumber}{}
\makeatletter
\newenvironment{assumption}[1]
 {%
  \renewcommand{\assumptionnumber}{Assumption #1}%
  \begin{assumption*}%
  \protected@edef\@currentlabel{#1}%
 }
 {%
  \end{assumption*}
 }
\makeatother

%% file: main.bbl
\newcommand{\etalchar}[1]{$^{#1}$}
\begin{thebibliography}{DMM{\etalchar{+}}17}

\bibitem[AM07]{Anderson2007-mn}
Brian D~O Anderson and John~B Moore.
\newblock {\em Optimal {C}ontrol: {L}inear {Q}uadratic {M}ethods}.
\newblock Dover, February 2007.

\bibitem[BB08]{basar2008h}
Tamer Ba{\c{s}}ar and Pierre Bernhard.
\newblock {\em $H_{\infty}$ Optimal Control and related Minimax Design
  Problems: a Dynamic Game Approach}.
\newblock Springer Science \& Business Media, 2008.

\bibitem[Ber05]{bertsekas2005dynamic}
D.P. Bertsekas.
\newblock {\em Dynamic Programming and Optimal Control}.
\newblock Number v. 2 in Athena Scientific optimization and computation series.
  Athena Scientific, 2005.

\bibitem[Ber13]{Bernhard2013-aa}
Pierre Bernhard.
\newblock Linear quadratic {Zero-Sum} {Two-Person} differential games, 2013.

\bibitem[BIM12]{bini2012numerical}
Dario~A Bini, Bruno Iannazzo, and Beatrice Meini.
\newblock {\em Numerical Solution of Algebraic Riccati Equations}, volume~9.
\newblock SIAM, 2012.

\bibitem[BMFM19]{bu2019lqr}
Jingjing Bu, Afshin Mesbahi, Maryam Fazel, and Mehran Mesbahi.
\newblock {LQR} through the lens of first order methods: {D}iscrete-time case.
\newblock {\em arXiv preprint arXiv:1907.08921}, 2019.

\bibitem[BMM19]{bu2019topological}
Jingjing Bu, Afshin Mesbahi, and Mehran Mesbahi.
\newblock On topological and metrical properties of stabilizing feedback gains:
  the {MIMO} case.
\newblock {\em arXiv preprint arXiv:1904.02737}, 2019.

\bibitem[BO99]{basar1999dynamic}
Tamer Basar and Geert~Jan Olsder.
\newblock {\em Dynamic Noncooperative Game Theory}, volume~23.
\newblock SIAM, 1999.

\bibitem[BR95]{bernhard1995theorem}
Pierre Bernhard and Alain Rapaport.
\newblock On a theorem of {D}anskin with an application to a theorem of von
  {N}eumann-{S}ion.
\newblock {\em Nonlinear Analysis}, 24(8):1163--1182, 1995.

\bibitem[DMM{\etalchar{+}}17]{dean2017sample}
Sarah Dean, Horia Mania, Nikolai Matni, Benjamin Recht, and Stephen Tu.
\newblock On the sample complexity of the linear quadratic regulator.
\newblock {\em arXiv preprint arXiv:1710.01688}, 2017.

\bibitem[Eng05]{Engwerda2005-np}
Jacob Engwerda.
\newblock {\em {LQ} Dynamic Optimization and Differential Games}.
\newblock John Wiley \& Sons, November 2005.

\bibitem[FGKM18]{fazel2018global}
Maryam Fazel, Rong Ge, Sham Kakade, and Mehran Mesbahi.
\newblock Global convergence of policy gradient methods for the linear
  quadratic regulator.
\newblock In {\em Proceedings of the 35th International Conference on Machine
  Learning}, pages 1467--1476, 2018.

\bibitem[{Hew}71]{hewer1971iterative}
G.~{Hewer}.
\newblock An iterative technique for the computation of the steady state gains
  for the discrete optimal regulator.
\newblock {\em IEEE Transactions on Automatic Control}, 16(4):382--384, 1971.

\bibitem[JCD{\etalchar{+}}19]{jaderberg2019human}
Max Jaderberg, Wojciech~M Czarnecki, Iain Dunning, Luke Marris, Guy Lever,
  Antonio~Garcia Castaneda, Charles Beattie, Neil~C Rabinowitz, Ari~S Morcos,
  Avraham Ruderman, et~al.
\newblock Human-level performance in 3d multiplayer games with population-based
  reinforcement learning.
\newblock {\em Science}, 364(6443):859--865, 2019.

\bibitem[LR95]{Lancaster1995algebraic}
Peter Lancaster and Leiba Rodman.
\newblock {\em Algebraic {R}iccati {E}quations}.
\newblock Oxford University Press, New York, {NY}, 1995.

\bibitem[MKS{\etalchar{+}}15]{mnih2015human}
Volodymyr Mnih, Koray Kavukcuoglu, David Silver, Andrei~A Rusu, Joel Veness,
  Marc~G Bellemare, Alex Graves, Martin Riedmiller, Andreas~K Fidjeland, Georg
  Ostrovski, et~al.
\newblock Human-level control through deep reinforcement learning.
\newblock {\em Nature}, 518(7540):529, 2015.

\bibitem[SHM{\etalchar{+}}16]{silver2016mastering}
David Silver, Aja Huang, Chris~J Maddison, Arthur Guez, Laurent Sifre, George
  Van Den~Driessche, Julian Schrittwieser, Ioannis Antonoglou, Veda
  Panneershelvam, Marc Lanctot, et~al.
\newblock Mastering the game of go with deep neural networks and tree search.
\newblock {\em nature}, 529(7587):484, 2016.

\bibitem[SLZ{\etalchar{+}}18]{srinivasan2018actor}
Sriram Srinivasan, Marc Lanctot, Vinicius Zambaldi, Julien P{\'e}rolat, Karl
  Tuyls, R{\'e}mi Munos, and Michael Bowling.
\newblock Actor-critic policy optimization in partially observable multiagent
  environments.
\newblock In {\em Advances in Neural Information Processing Systems}, pages
  3422--3435, 2018.

\bibitem[Sto90]{stoorvogel1990h}
Anton Stoorvogel.
\newblock {\em The ${H}_{\infty}$ Control Problem: a State Space Approach}.
\newblock Citeseer, 1990.

\bibitem[SW94]{stoorvogel1994discrete}
Anton~A Stoorvogel and Arie~JTM Weeren.
\newblock The discrete-time {R}iccati equation related to the ${H}_{\infty}$
  control problem.
\newblock {\em IEEE Transactions on Automatic Control}, 39(3):686--691, 1994.

\bibitem[VMKL17]{Vamvoudakis2017-ls}
Kyriakos~G Vamvoudakis, Hamidreza Modares, Bahare Kiumarsi, and Frank~L Lewis.
\newblock Game {Theory-Based} control system algorithms with {Real-Time}
  reinforcement learning: How to solve multiplayer games online.
\newblock {\em IEEE Control Syst.}, 37(1):33--52, 2017.

\bibitem[Wil71]{willems1971least}
Jan Willems.
\newblock Least squares stationary optimal control and the algebraic {R}iccati
  equation.
\newblock {\em IEEE Transactions on Automatic Control}, 16(6):621--634, 1971.

\bibitem[Won12]{wonham2012linear}
W.M. Wonham.
\newblock {\em Linear Multivariable Control: a Geometric Approach: A Geometric
  Approach}.
\newblock Stochastic Modelling and Applied Probability. Springer New York,
  2012.

\bibitem[Zha05]{Zhang2005-pe}
Pingjian Zhang.
\newblock Some results on {Two-Person} {Zero-Sum} linear quadratic differential
  games.
\newblock {\em SIAM J. Control Optim.}, 43(6):2157--2165, January 2005.

\bibitem[ZYB19]{zhang2019policy}
Kaiqing Zhang, Zhuoran Yang, and Tamer Ba{\c{s}}ar.
\newblock Policy optimization provably converges to {N}ash equilibria in
  zero-sum linear quadratic games.
\newblock {\em arXiv preprint arXiv:1906.00729}, 2019.

\end{thebibliography}
